\documentclass[aps,prl,reprint,floatfix,superscriptaddress,longbibliography]{revtex4-2}
\usepackage{amsmath,amssymb,amsthm,mathtools}
\usepackage{bm}
\usepackage{booktabs}
\usepackage{graphicx}
\usepackage{microtype}
\usepackage[colorlinks=true,linkcolor=blue,citecolor=blue,urlcolor=blue]{hyperref}

\newtheoremstyle{prlrunin}{\topsep}{\topsep}{\normalfont}{\parindent}
  {\itshape}{---}{0.3em}{}
\theoremstyle{prlrunin}
\newtheorem{theorem}{Theorem}
\newtheorem{corollary}[theorem]{Corollary}
\newtheorem{lemma}[theorem]{Lemma}
\newcounter{algorithm}

\newcommand{\Tr}{\operatorname{Tr}}
\newcommand{\cE}{\mathcal E}
\newcommand{\cH}{\mathcal H}
\newcommand{\Iacc}{I_{\mathrm{acc}}}
\newcommand{\Uaff}{U_{\mathrm{aff}}}
\newcommand{\supp}{\operatorname{supp}}
\newcommand{\norm}[1]{\lVert#1\rVert}

\newcommand{\papertitle}{Information-Optimal Quantum Measurements: Posterior Algebra and the Limits of Orthogonal Readout}

\hypersetup{
  pdftitle={\papertitle},
  pdfauthor={Jinbo Wang, Qihang Wang, Pengcheng Hou, and Kun Chen},
  pdfkeywords={accessible information, posterior algebra, Shor conjecture, optimal quantum measurement, information-error trade-off}
}

\newcommand{\prlappendix}[1]{%
  \par\refstepcounter{section}\setcounter{equation}{0}%
  \textit{\thesection.~#1}---\ignorespaces
}
\newcommand{\prlsubheading}[1]{\paragraph*{\normalfont #1:}}

\newcommand{\dd}{\mathrm d}
\begin{document}
\providecommand{\Pe}{P_{\mathrm e}}
\providecommand{\VPOVM}{V_{\mathrm{POVM}}}
\providecommand{\VPVM}{V_{\mathrm{PVM}}}

\title{\papertitle}

\author{Jinbo Wang}
\thanks{These authors contributed equally to this work.}
\affiliation{School of Mathematical Sciences, Peking University, Beijing 100871, China}
\author{Qihang Wang}
\thanks{These authors contributed equally to this work.}
\affiliation{School of Mathematical Sciences, Peking University, Beijing 100871, China}
\author{Pengcheng Hou}
\email[Corresponding author: ]{houpc@hfnl.cn}
\affiliation{Hefei National Laboratory, University of Science and Technology of China, Hefei 230088, China}
\author{Kun Chen}
\email[Corresponding author: ]{chenkun@itp.ac.cn}
\affiliation{Institute of Theoretical Physics, Chinese Academy of Sciences, Beijing 100190, China}

\begin{abstract}
When do orthogonal measurements extract all accessible information? We introduce a posterior algebra for full-rank ensembles. Its commutativity determines whether an information bound is attained; its joint eigenspaces then classify optimal finite-outcome measurements. For two states, we constructively prove Shor's conjecture in finite dimension, including singular states. Orthogonal measurements can fail at fixed guessing error: a discontinuous basis switch forces some qubit pairs to need three outcomes. A Tsallis information score can require three outcomes even without an error constraint.
\end{abstract}
\maketitle

\textit{Introduction}---A quantum measurement converts a signal state into a classical record. For a known ensemble, the receiver uses this record to guess which state was prepared and assess confidence in that guess. Two outcomes may favor the same input state yet assign it different posterior probabilities. Merging them preserves the best guess but can discard information useful for subsequent decoding [Fig.~\ref{fig:schematic}(a)]. Such soft information improves decoding in quantum error correction~\cite{pattison2021soft,ali2024soft}. The retained information is the mutual information between the preparation label and the outcome. Its maximum over measurements, the accessible information, obeys Holevo's upper bound~\cite{holevo1973bounds} and further ensemble-dependent bounds~\cite{fuchs1994ensemble,jain2007accessible}. 

\begin{figure}[!t]
    \centering
    \includegraphics[width=\columnwidth]{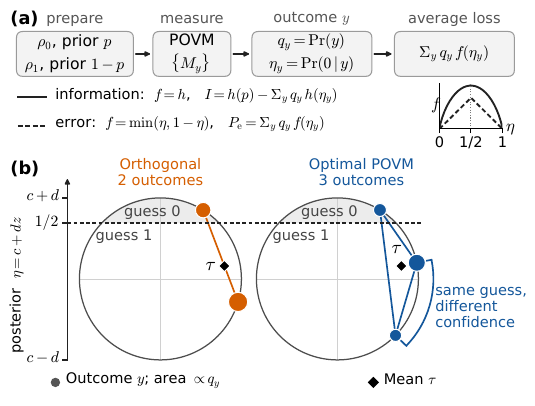}
    \caption{A third outcome can preserve confidence within a guess.
        (a) A positive operator-valued measure (POVM) $\{M_y\}$ gives outcome $y$
        with probability $q_y$ and posterior $\eta_y=\Pr(0|y)$.
        Mutual information $I$ uses binary entropy $h$ (solid); guessing error
        $\Pe$ uses $\min(\eta,1-\eta)$ (dashed). Merging distinct posteriors
        on the same side of $1/2$ preserves error but loses information.
        (b) Three outcomes (right) retain two confidence levels within guess 1;
        the orthogonal readout shown (left) separates the guesses.
        Dots represent transformed states
        $|\psi_y\rangle\langle\psi_y|=\tau^{1/2}M_y\tau^{1/2}/q_y$
        of rank-one effects, with areas proportional to $q_y$ and weighted mean
        equal to the average signal state $\tau$ ($\blacklozenge$).
        In each Bloch-sphere section, height $z\in[-1,1]$ gives $\eta=c+dz$, with midpoint $c$ and half-span $d$;
        shading marks guess 0. The orthogonal points form a chord through
        $\tau$ and need not be antipodal. Both readouts optimize the same
        qubit pair at a common error-penalty weight; their errors may differ.}
    \label{fig:schematic}
\end{figure}

When the signal states commute, measuring their joint eigenbasis extracts all accessible information~\cite{suzuki2007accessible}. For two signals, an orthogonal measurement also minimizes the guessing error~\cite{helstrom1976quantum}. Shor conjectured that orthogonal measurements likewise attain the accessible information of any two states~\cite{shor2000number,shor2000gafa}, although the two optimal bases generally differ~\cite{fuchs1995distinguishability,varga2024}. After proofs for pure states and for qubits~\cite{levitin1995optimal,keil2024}, the conjecture has recently been established in arbitrary finite dimension through measured-divergence methods~\cite{fang2026uhlmann} and a variational construction~\cite{jo2026shor}. Beyond such cases, existing criteria test whether a given candidate measurement is globally optimal~\cite{holevo1973information,holevo2022optimization,holevo2026accessible,salazar2026benign}. Orthogonal measurements need not suffice for larger ensembles: even three symmetric qubit states can require more outcomes than the Hilbert-space dimension~\cite{sasakietal1999accessible,shor2000number}.

Two questions remain. Which measurement outcomes must remain distinguishable to extract all accessible information? And does orthogonal sufficiency survive when the same readout must also make reliable decisions? Information and guessing probability need not even be monotonically related~\cite{thai2026shor}. Because all posteriors must arise from a single measurement, it is unclear whether orthogonal readouts, even randomized ones, attain the best trade-off between them.

In this Letter, we answer these questions through the structure of measurement posteriors. First, the posterior algebra identifies which outcomes must remain distinguishable to maximize information. For a full-rank ensemble, maximizing a concave objective gives an upper bound on accessible information and a set of optimized posterior operators. The algebra they generate is commutative exactly when the bound is attainable. In that case, their joint eigenspaces classify all optimal finite-outcome measurements: different eigenspaces must remain distinguishable, whereas further distinctions within an eigenspace carry no additional information. The criterion contains the commuting-signal case and extends it to noncommuting signals. For two distinct states with positive priors, a single optimized observable gives this classification even for singular states. Together with the trivial cases of identical states or a vanishing prior, this yields a constructive proof of Shor's conjecture~\footnote{An earlier version containing our proof of Shor's conjecture for arbitrary finite-dimensional binary ensembles, including singular states, was publicly archived on 13 August 2026 (\url{https://doi.org/10.5281/zenodo.21911604}). Our work was developed independently of Ref.~\cite{jo2026shor}.}.

Second, we establish limits of orthogonal readout when the task changes. We show that suitable full-rank qubit pairs need three outcomes to extract more information at equal error than every randomized orthogonal strategy. A discontinuous switch of the best orthogonal basis as the error penalty increases forces a strict three-outcome advantage: the extra outcome retains different confidence levels within the same guess [Fig.~\ref{fig:schematic}(b)]. In our example, the gain reaches $3.1\%$ of the accessible information. Changing the information score to a Tsallis entropy~\cite{tsallis1988possible} can make three outcomes necessary even without an error penalty.
\par
\textit{Posterior algebra and optimal measurements}---Consider a finite-dimensional ensemble $\mathcal E=\{(p_i,\rho_i)\}_{i=1}^m$, where the source prepares state $\rho_i$ with prior probability $p_i>0$, $\sum_i p_i=1$, and $\tau=\sum_i p_i\rho_i$ is the average state. All operators act on $\operatorname{supp}\tau$; full rank means positive definiteness on this support. A positive operator-valued measure (POVM) $M=\{M_y\}$ consists of positive effects satisfying $\sum_yM_y=I$; an orthogonal measurement is a projection-valued measure (PVM), with mutually orthogonal projectors. The conditional outcome probability is $q(y|i)=\operatorname{Tr}(\rho_iM_y)$ and its average over preparations is $q_y=\operatorname{Tr}(\tau M_y)$. Omitting outcomes that never occur, the mutual information and accessible information are
\begin{equation}
 I(M)=\sum_{i,y}p_iq(y|i)\log\frac{q(y|i)}{q_y},
 \qquad I_{\mathrm{acc}}=\max_M I(M).
 \label{eq:letter-information}
\end{equation}
Logarithms are natural unless bits are stated. The posterior probability after outcome $y$ is $p_i g_{iy}$, where $g_{iy}=q(y|i)/q_y$ is the posterior-to-prior ratio.

For full-rank ensembles, we construct an operator optimization that preserves the affine relations among the signal states. Choose a real basis $\{A_a\}_{a=1}^r$ of the span of $\rho_i-\tau$ and write $\rho_i=\tau+\sum_a s_{ia}A_a$, with $\sum_i p_i s_{ia}=0$. Introduce a tuple of Hermitian variables $\bm H=(H_1,\ldots,H_r)$ and affine operators $G_i(\bm H)=I+\sum_a s_{ia}H_a$. A physical measurement induces $G_i(\bm H_M)=\sum_y g_{iy}M_y$: these are operator-valued averages of classical posterior-to-prior ratios, rather than posterior density operators. We drop the common-POVM requirement but retain the affine relations and positive definiteness, $G_i(\bm H)\succ0$. Every measurement gives a feasible tuple, but not every feasible tuple comes from a measurement. We maximize
\begin{equation}
 \Phi(\bm H)=\sum_i p_i\operatorname{Tr}\rho_i\log G_i(\bm H),
 \qquad U_{\mathrm{aff}}=\max_{\bm H}\Phi(\bm H).
 \label{eq:v3-affine-program}
\end{equation}
The feasible domain is bounded, and full rank makes $\Phi$ strictly concave with $\Phi\to-\infty$ at its boundary. Hence there is a unique interior maximizer $\bm H_\star$. We call $G_i^\star=G_i(\bm H_\star)$ the optimized posterior operators and the algebra they generate the posterior algebra; both are independent of the chosen basis $\{A_a\}$ (see Sec.~\ref{app:v3-geometry} of the Supplemental Material (SM)~\cite{supplemental}).

The operator Jensen inequality~\cite{hansen2003jensen} gives $\log G_i(\bm H_M)\succeq\sum_y(\log g_{iy})M_y$, and hence
\begin{equation}
 I(M)\leq\Phi(\bm H_M)\leq U_{\mathrm{aff}}.
 \label{eq:letter-upper-bound}
\end{equation}
Our central result is that, for full-rank ensembles, this affine bound is attained if and only if the optimized posterior operators commute:
\begin{equation}
 I_{\mathrm{acc}}=U_{\mathrm{aff}}
 \quad\Longleftrightarrow\quad
 [G_i^\star,G_j^\star]=0\quad\text{for all }i,j.
 \label{eq:letter-commutation}
\end{equation}
Commuting posterior operators can be resolved simultaneously. Their joint-spectral PVM $P_\star=\{P_z\}$ projects onto the common eigenspaces for distinct eigenvalue tuples. The equality analysis shows that this measurement attains the bound and that optimal finite-outcome POVMs are precisely its refinements, with outcomes partitioned into groups $Y_z$ satisfying
\begin{equation}
 M_y=P_zM_yP_z\quad(y\in Y_z),
 \qquad \sum_{y\in Y_z}M_y=P_z.
 \label{eq:v3-refinement}
\end{equation}
Thus, distinct joint eigenspaces must remain distinguishable, whereas refinements within one eigenspace carry no additional information. For this ensemble, every such refinement and $P_\star$ simulate each other's outcome statistics by classical postprocessing independent of the input label. This is Blackwell equivalence~\cite{blackwell1953equivalent,kuramochi2018experiments} (SM, Secs.~\ref{app:v3-geometry}--\ref{app:v3-exactness}).

\begin{figure}[tb]
 \centering
 \includegraphics[width=\columnwidth]{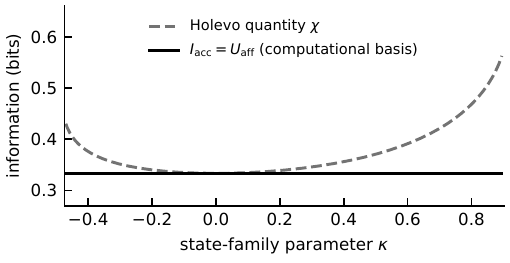}
 \caption{Commuting posterior operators select an optimal basis for
 noncommuting signals. The three equally likely qutrit states
 $\rho_i(\kappa)$ [Eq.~\eqref{eq:em-qutrit} and cyclic permutations] have
 diagonal optimized posterior operators [Eq.~\eqref{eq:em-qutrit-G}]
 throughout the full-rank interval $\kappa_-<\kappa<\kappa_+$.
 The computational basis attains the accessible information
 $I_{\rm acc}=U_{\rm aff}=1/3$ bit, saturating the affine bound.
 For $\kappa\neq0$, the states do not commute and the Holevo quantity
 satisfies $\chi>I_{\rm acc}$ (End Matter~\ref{em:example}).}
 \label{fig:posterior-algebra}
\end{figure}

If the bound is attained, uniqueness identifies the optimal measurement's induced tuple with $\bm H_\star$. Jensen equality and full rank then place each effect within a common eigenspace of the $G_i^\star$. Completeness makes these supports span the space, forcing commutation. A Gibbs-inequality identity proves sufficiency, and stationarity gives the complete classification in Eq.~\eqref{eq:v3-refinement} (End Matter~\ref{em:exactness}). Unlike candidate-based optimality conditions, the criterion is decided from the ensemble alone. It concerns attainment of $U_{\mathrm{aff}}$: an orthogonal optimum can also exist when the bound is strict (SM, Sec.~\ref{app:v3-scope}).

The qutrit family in Fig.~\ref{fig:posterior-algebra} varies the off-diagonal coherences without changing the computational-basis probabilities. For $\kappa\neq0$, the three signal states do not commute pairwise, yet their optimized posterior operators commute. This certifies that the computational basis remains globally optimal throughout the full-rank interval. Noisy trine states instead have noncommuting posterior operators and a strict affine bound (End Matter~\ref{em:example}). Thus the criterion concerns the posterior operators, not the signal states.
\par
\textit{Two signals and Shor's conjecture}---For a binary ensemble $\mathcal E=\{(p,\rho_0),(1-p,\rho_1)\}$ with distinct states and priors $p,1-p$, where $0<p<1$, all state differences are collinear, so the affine optimization reduces to a single Hermitian observable $H$, with $G_0(H)=I+(1-p)H$ and $G_1(H)=I-pH$. The objective is
\begin{equation}
    \Phi_{\mathrm{bin}}(H)
    =p\operatorname{Tr}\rho_0\log G_0(H)
    +(1-p)\operatorname{Tr}\rho_1\log G_1(H),
    \label{eq:v3-binary-program}
\end{equation}
on the domain $G_0(H)\succ0$, $G_1(H)\succ0$. This variational form is affine-equivalent to Jo's Proposition~4.1~\cite{jo2026shor} (SM, Sec.~\ref{app:v3-binary}). For full-rank states, strict concavity and boundary divergence give a unique maximizer $H_\star$, whose optimized posterior operators are
\begin{equation}
    G_0^\star=I+(1-p)H_\star,
    \qquad
    G_1^\star=I-pH_\star.
    \label{eq:v3-binary-posteriors}
\end{equation}
Since both are functions of the same observable, they commute, and its spectral PVM attains accessible information by Eq.~\eqref{eq:letter-commutation}. The same construction applies to full-rank ensembles of mixtures of the same two states.

The eigenvalues of $H_\star$ have a direct operational meaning: an outcome with eigenvalue $x$ gives the posterior $\Pr(0|x)=p[1+(1-p)x]$. Thus, the spectral measurement resolves exactly the distinct posterior probabilities.

Singular binary states require a closed-domain extension on $\operatorname{supp}\tau$, with $-(1-p)^{-1}I\preceq H\preceq p^{-1}I$. At $H_\star$, boundary eigenvalues correspond to posterior probabilities 0 or 1: a spectral outcome then identifies the preparation with certainty. The extended logarithmic objective has a unique maximizer, including at the boundary. Boundary Jensen equality and stationarity within each spectral block give the same complete classification: the optimal finite-outcome POVMs are precisely the refinements of the spectral PVM of $H_\star$ (SM, Sec.~\ref{app:v3-binary}). Resolving each spectral block into an orthonormal basis therefore gives an orthogonal optimum for every finite-dimensional binary ensemble, including singular states, constructively proving Shor's conjecture.

This structure also controls approximate readout. For full-rank binary ensembles, a nearly optimal measurement has little total effect weight outside the dominant posterior block of each outcome. The control weakens as distinct posterior values merge or the states approach singularity (SM, Sec.~\ref{app:v3-robust}).

For full-rank binary ensembles, the same single-observable formulation gives a constructive receiver: optimize a feasible Hermitian $H$ and measure in its eigenbasis. For a trial $H$, let $P_H$ be its spectral PVM and $R(H)=\nabla\Phi_{\mathrm{bin}}(H)$. Strong concavity gives the computable interval (SM, Secs.~\ref{app:v3-robust}--\ref{app:v3-benchmark})
\begin{equation}
    I(P_H)\leq I_{\mathrm{acc}}
    \leq
    \Phi_{\mathrm{bin}}(H)
    +\frac{\operatorname{Tr}[R(H)^2]}{2\omega},
    \label{eq:letter-receiver-bound}
\end{equation}
where $\omega=p^2(1-p)^2[p\lambda_{\min}(\rho_0)+(1-p)\lambda_{\min}(\rho_1)]$ is a positive lower bound on the curvature of $-\Phi_{\mathrm{bin}}$. The gap bounds the information lost by implementing $P_H$; as an optimality certificate~\cite{coutts2021certifying}, it supplies a stopping criterion independent of the optimization algorithm.
\par
\textit{When orthogonal measurements cease to suffice}---Does an orthogonal measurement still suffice when we also constrain the guessing error? For binary classical channels with a fixed prior, the trade-off between mutual information and guessing probability has been characterized~\cite{thai2026shor}; here the quantum measurement constraint determines which posteriors are available. For a binary ensemble, let $\eta_y=\Pr(0|y)$ be the posterior probability of input 0, and let $h(x)=-x\log x-(1-x)\log(1-x)$ be the binary entropy. The information and the error of guessing the more likely input are
\begin{align}
 I(M)&=h(p)-\sum_y q_yh(\eta_y),\nonumber\\
 P_e(M)&=\sum_yq_y\min(\eta_y,1-\eta_y).
 \label{eq:letter-two-scores}
\end{align}
We optimize $F_\lambda(M)=I(M)-\lambda P_e(M)$ for $\lambda\geq0$, and denote its maxima over all POVMs and over PVMs by $V_{\mathrm{POVM}}(\lambda)$ and $V_{\mathrm{PVM}}(\lambda)$. Increasing $\lambda$ places greater weight on decision accuracy. For a random mixture of bases with the basis label retained, $F_\lambda=\sum_b w_bF_\lambda(P_b)\leq V_{\mathrm{PVM}}$, so this bound covers every such randomized strategy~\cite{oszmaniec2017simulating}. Define $\Delta_\lambda=V_{\mathrm{POVM}}-V_{\mathrm{PVM}}$. At the optimal POVM's error budget, its information gain over randomized orthogonal readout is at least $\Delta_\lambda$.

\begin{figure}[tb]
    \centering
    \includegraphics[width=\columnwidth]{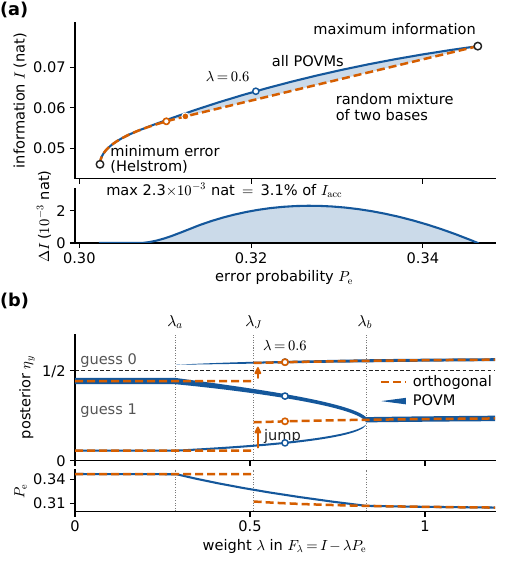}
    \caption{Strict information advantage over randomized orthogonal readout.
        (a) Maximum information $I$ at error budget $\Pe$: all POVMs (blue)
        and random mixtures of orthogonal measurements with the basis label
        retained (dashed vermilion). The straight segment from the filled
        point to the maximum-information endpoint mixes two bases.
        The lower strip shows $\Delta I$, blue minus dashed at equal error.
        Black open circles mark the endpoint optima.
        (b) Posteriors $\eta_y$ (top) and errors $\Pe$ (bottom) of measurements
        maximizing $F_\lambda=I-\lambda\Pe$: a single orthogonal basis
        (dashed) and a POVM (blue). The basis jumps at $\lambda_J$; the POVM
        varies continuously and has three outcomes for
        $\lambda_a<\lambda<\lambda_b$, two giving the same guess at different confidence. Blue line thickness in the upper panel indicates outcome
        probabilities $q_y$. Randomization does not improve the weighted
        orthogonal optimum. Colored open circles mark $\lambda=0.6$, the
        common penalty weight used in Fig.~\ref{fig:schematic}(b), at generally different
        errors. The ensemble has
        $Q=\tau^{-1/2}p\rho_0\tau^{-1/2}=0.30I+0.29\sigma_z$ and
        $\tau=(I+0.79\sigma_x+0.16\sigma_z)/2$, with $p=0.3464$;
        $p\rho_0=\tau^{1/2}Q\tau^{1/2}$ and
        $(1-p)\rho_1=\tau-p\rho_0$.}
    \label{fig:pareto-gap}
\end{figure}

To determine the required number of outcomes, consider two full-rank qubit states and write $Q=\tau^{-1/2}p\rho_0\tau^{-1/2}=cI+d\sigma_z$, with $d>0$; this fixed operator describes the posteriors available to all measurements. Refinement cannot lower $F_\lambda$, so rank-one effects suffice. Each effect gives a pure qubit state through
\begin{equation}
 \tau^{1/2}M_y\tau^{1/2}=q_y|\psi_y\rangle\langle\psi_y|,
 \qquad \sum_yq_y|\psi_y\rangle\langle\psi_y|=\tau.
 \label{eq:letter-posterior-geometry}
\end{equation}
The transformed pure states therefore average to $\tau$ [Fig.~\ref{fig:schematic}(b)]. In the eigenbasis of $Q$, choose the transverse axis so that $\tau$ has Bloch vector $(x_\tau,0,z_\tau)$, with $x_\tau\geq0$. An outcome's height $z_y$ gives its posterior $\eta_y=c+dz_y$, with extremes $c\pm d$. Its horizontal coordinate cannot exceed $s(z_y)=\sqrt{1-z_y^2}$, so completeness requires $\sum_yq_ys(z_y)\geq x_\tau$.

The single-outcome score $-h(\eta)-\lambda\min(\eta,1-\eta)$ is strictly convex. If the inequality is strict, mix the height distribution with one on the poles $z=\pm1$ having the same mean height. The poles have $s=0$, so a suitable mixture reduces $\sum_yq_ys(z_y)$ to $x_\tau$ while preserving mean height and increasing the mean score. This mixture is realizable on the positive semicircle $(s(z),0,z)$ with the original mean state $\tau$. The constraint therefore saturates at an optimum; since no point can exceed its horizontal bound, equality of the averages forces every point onto this semicircle. The averages of $1,z,s(z)$ fix normalization, mean height, and mean horizontal coordinate. The objective is linear in the outcome probabilities, so an extreme optimal distribution has at most three outcomes (SM, Secs.~\ref{app:pareto-scalar}--\ref{em:three}).

Conversely, every two-outcome qubit POVM $\{E,I-E\}$ is a classical postprocessing of the spectral PVM of $E$. Such postprocessing cannot increase information or improve guessing success. Thus, two outcomes cannot give a strict advantage, whereas three suffice: whenever $\Delta_\lambda>0$, the minimal number of outcomes needed for an optimal measurement is three. A rank-one three-outcome qubit POVM admits a projective realization in a three-dimensional extension~\cite{peres1993quantum}.

A sufficient mechanism for $\Delta_\lambda>0$ is a discontinuous switch of the optimal orthogonal basis. Consider noncommuting full-rank qubit states for which the extreme posteriors straddle $1/2$, but both outcomes of the information-optimal PVM at $\lambda=0$ strictly favor the same preparation. This measurement distinguishes confidence without changing the guess. Its error is $\min(p,1-p)$ and is constant under sufficiently small basis rotations. Such rotations reduce information without changing the penalty, so the original basis remains locally optimal as $\lambda$ grows. The Helstrom measurement has strictly smaller error and eventually outperforms it in $F_\lambda$. The maximizing basis cannot rotate continuously away from this locally optimal basis: it must jump, with two distinct PVMs tying at some $\lambda_J>0$ [Fig.~\ref{fig:pareto-gap}(b)].

A jump forces a strict advantage because the unrestricted optimal posterior distribution is unique (SM, Lemma~\ref{lem:pareto-unique}), fixing the effects up to proportional splitting. If either tied PVM were optimal over all POVMs, both would be, contradicting uniqueness. A three-outcome POVM consequently beats every PVM, and hence every randomized orthogonal strategy, over an interval containing the switch. This mechanism is sufficient but not necessary (SM, Theorem~\ref{thm:pareto-forced}; Secs.~\ref{em:three} and~\ref{app:pareto-switch}).

The same geometry determines which posterior spectra allow a three-outcome advantage. For $Q=cI+d\sigma_z$, with $0<c-d<c+d<1$, there exists a compatible full-rank average state $\tau$ and a penalty $\lambda>0$ giving a strict three-outcome advantage if and only if
\begin{equation}
    c-d<\frac12<c+d,
    \qquad
    c\neq\frac12.
    \label{eq:letter-existence}
\end{equation}
The extreme posteriors must favor opposite guesses and form a range asymmetric about $1/2$. At fixed $Q$, we vary $\tau$ and reconstruct the prior as $p=\operatorname{Tr}(\tau Q)$; this spectral asymmetry is compatible with equal priors. The criterion guarantees an advantage for some, but not necessarily every, compatible $\tau$ (SM, Sec.~\ref{app:pareto-construct}).

Figure~\ref{fig:pareto-gap}(a) compares information at a common error budget. The illustrated ensemble gains up to $2.3\times10^{-3}$ nat, or $3.1\%$ of the accessible information, over randomized orthogonal strategies. Convex duality equates the largest fixed-error information gain to $\max_\lambda\Delta_\lambda$ (SM, Sec.~\ref{app:extension-overview}). Here the weighted advantage occurs for $\lambda_a=0.288<\lambda<\lambda_b=0.832$, around the basis switch at $\lambda_J=0.510$ (SM, Secs.~\ref{app:pareto-switch}--\ref{app:pareto-size}).

For fixed input composition and independent single-system readout, this advantage also yields an achievable communication-rate gain at equal average per-symbol error (SM, Sec.~\ref{app:pareto-physical}).

The required distinctions also depend on the information score. Replace $h$ in Eq.~\eqref{eq:letter-two-scores} by the Tsallis entropy~\cite{tsallis1988possible} of index $1/4$, $f(\eta)=\frac{4}{3}[\eta^{1/4}+(1-\eta)^{1/4}-1]$, and set $\lambda=0$. A certified full-rank binary ensemble then requires three outcomes to maximize the resulting generalized information score (SM, Sec.~\ref{app:tsallis-witness}). This full-rank example also shows that operator convexity of the divergence generator alone does not ensure projective sufficiency~\cite{hiai2017different}; the connection and the earlier pure-state examples~\cite{katz2024optimal} are discussed in SM, Sec.~\ref{app:operator-convex-score}.
\par
\textit{Conclusion}---The posterior algebra identifies which measurement outcomes must remain distinguishable when the affine information bound is attained. For full-rank ensembles, commutativity characterizes attainability of the affine bound, and the joint eigenspaces then classify all optimal finite-outcome measurements, including for noncommuting signals. For two distinct states with positive priors, a single posterior observable also classifies the optimal measurements for singular states, giving a constructive proof of Shor's conjecture. Orthogonal sufficiency is nevertheless a property of the readout task, not just of the number of signal states. At fixed guessing error, suitable full-rank qubit pairs gain information from a third outcome that preserves confidence within a guess, and a Tsallis score can require three outcomes even without an error constraint. Whether the advantage always occupies a single interval of error penalties remains open. Optical implementations of generalized measurements~\cite{clarke2001experimental,mohseni2004optical} provide a route to testing the predicted advantage.
\par
\begin{acknowledgments}
\textit{Acknowledgments}---We acknowledge GeWu SciLab for providing the collaborative
research environment in which this project was developed.
K.C. acknowledges support from the Strategic Priority Research Program
of the Chinese Academy of Sciences under Grant No.~XDB1680102.

\textit{Disclosure of automated assistance}---OpenAI GPT-5.6 Sol,
OpenAI GPT-6 Astra, OpenAI GPT-6.1 Sol, Anthropic Claude Fable 5,
Anthropic Claude Fable 5.1, Anthropic Claude Opus 5.5, and xAI Grok 4.6
were used in developing and preparing this manuscript. The AI-assisted
workflow developed the posterior algebra and optimal-measurement arguments;
checked derivations, equality conditions, singular limits, stability bounds,
and receiver certificates; examined information--error trade-offs, basis switches, and generalized
information scores; reviewed numerical protocols and results;
checked references; and prepared the manuscript and figures.
The authors defined the research objectives, reviewed the AI-assisted
derivations and citations, and revised the manuscript. The authors take
full responsibility for its content.
\end{acknowledgments}

\textit{Data availability}---The numerical code, fixed inputs, per-case
results, figure data and scripts, interval certificates, and reproduction
instructions are available in the Zenodo reproducibility archive,
version 1.0~\cite{wang2026reproducibility}.
The archive contains the 408 benchmark inputs, generated with the protocol
of Sec.~\ref{app:v3-benchmark}, and their results; the benchmark script
regenerates the intermediate measurement matrices.

\bibliography{references_prl}
\clearpage
\makeatletter
\let\regularsectionheading\@hangfrom@section
\let\regularsectionnumber\@sectioncntformat
\makeatother
\appendix
\section*{End Matter}

\prlappendix{Posterior-algebra criterion}
\label{em:exactness}%
We give the equality argument for the posterior-algebra criterion,
Eqs.~\eqref{eq:letter-commutation} and~\eqref{eq:v3-refinement}.
The states are full rank throughout. The upper bound follows from
Eq.~\eqref{eq:letter-upper-bound}; complete proofs are in
Secs.~\ref{app:v3-geometry} and~\ref{app:v3-exactness} of the SM.

\textit{Sufficiency.} If the $G_i^\star$ commute, their joint
eigenprojectors $P_\star=\{P_z\}$, on which $G_i^\star$ takes the values
$g_{iz}$, satisfy
\begin{equation}
 I(P_\star)-\Phi(\bm H_\star)
 =\sum_z\Tr(\tau P_z)D(\alpha_{\cdot z}\Vert\beta_{\cdot z})\geq0,
 \label{eq:v3-gap-identity}
\end{equation}
with $\alpha_{iz}=p_i\Tr(\rho_iP_z)/\Tr(\tau P_z)$ the measured and
$\beta_{iz}=p_ig_{iz}$ the predicted posteriors. Here
$D(\alpha\Vert\beta)=\sum_i\alpha_i\log(\alpha_i/\beta_i)$ is classical
relative entropy; its nonnegativity expresses the properness of the logarithmic score. Hence
$I(P_\star)\geq\Phi(\bm H_\star)=\Uaff$, so $P_\star$ is optimal
and $\Iacc=\Uaff$.

\textit{Necessity.} If $\Iacc=\Uaff$, an optimal finite-outcome POVM $M$ exists
\cite{davies1978information}, and
$I(M)=\Phi(\bm H_M)=\Phi(\bm H_\star)$; uniqueness of the
maximizer gives $\bm H_M=\bm H_\star$. Write the POVM as the isometry
$Vx=\bigoplus_yM_y^{1/2}x$ and put $T_i=\bigoplus_yg_{iy}I$, so that
$G_i^\star=V^*T_iV$. For such a compression $G=V^*TV$, write
$T=\left(\begin{smallmatrix}G&B\\B^*&N\end{smallmatrix}\right)$ on
$\operatorname{ran}V\oplus(\operatorname{ran}V)^\perp$; then
\begin{equation}
 \log G-V^*(\log T)V=\int_0^\infty(S_u^{-1}-G_u^{-1})\,\dd u,
 \label{eq:em-schur}
\end{equation}
with $G_u=G+uI$, $N_u=N+uI$, and $S_u=G_u-BN_u^{-1}B^*$; this gap is
positive semidefinite, and it
vanishes only if $B=0$, that is,
$TV=VG$. Because each $\rho_i$ is full rank, a zero trace of this positive gap forces
the gap itself to vanish, so
\begin{equation}
 g_{iy}M_y^{1/2}=M_y^{1/2}G_i^\star\qquad\text{for all }i,y.
 \label{eq:em-intertwining}
\end{equation}
Every effect is therefore supported in one joint eigenspace of all
$G_i^\star$, and since the effects sum to $I$, the $G_i^\star$ commute.

\textit{Classification.} Each $H_a^\star$ is a real
linear combination of the $G_i^\star-I$, hence scalar on every joint
eigenspace $P_z$. Compressing the stationarity conditions of $\Phi$ to $P_z$
gives
\begin{equation}
 P_z\rho_iP_z=g_{iz}P_z\tau P_z.
 \label{eq:em-proportional}
\end{equation}
Any refinement as in Eq.~\eqref{eq:v3-refinement} thus has
$\Tr\rho_iM_y=g_{iz}\Tr\tau M_y$ for $y\in Y_z$: its outcome statistics
on this ensemble arise from $P_\star$ through
$W(y|z)=\Tr(\tau M_y)/\Tr(\tau P_z)$, which does not depend on the input
label $i$. The two induced classical experiments are therefore Blackwell
equivalent (each can be simulated from the other by classical postprocessing)
\cite{blackwell1953equivalent,kuramochi2018experiments}, and the refinement is
optimal. Conversely, Eq.~\eqref{eq:em-intertwining} places every optimal effect in
one block, so no other finite-outcome POVM is optimal. For two full-rank
states and rank-one PVMs, the same blockwise proportionality is the
global-optimality criterion for critical PVMs in Theorems~4.1 and~4.5 of
Ref.~\cite{salazar2026benign}.

\prlappendix{Examples with three states}
\label{em:example}%
Two examples illustrate the posterior-algebra criterion: noisy trines, for which the bound is not attained, and three
noncommuting qutrit states whose posterior operators commute.

\textit{States that commute with $\tau$.} If every full-rank $\rho_i$
commutes with $\tau$, the choice $H_a=\tau^{-1}A_a$ is Hermitian and gives
$G_i=\tau^{-1}\rho_i\succ0$ and
$\Phi=\sum_ip_i\Tr\rho_i(\log\rho_i-\log\tau)=\chi$. Since $\Uaff\leq\chi$,
this is the maximizer: $G_i^\star=\tau^{-1}\rho_i$ and $\Uaff=\chi$. These
operators commute exactly when the states do. For a noisy trine,
$\rho_i=(I+\gamma\,\bm n_i\cdot\bm\sigma)/2$ with equal priors,
$0<\gamma<1$, and coplanar unit vectors $\bm n_i$ at $120^\circ$, one has
$\tau=I/2$, so $G_i^\star=2\rho_i$ do not commute and $\Iacc<\Uaff$.

\textit{A qutrit family.} Take three qutrit states with $p_i=1/3$, $\ell=3/(4\ln2)$, and
\begin{equation}
 \rho_1(\kappa)=\frac16\begin{pmatrix}4&2\ell \kappa&2\ell \kappa\\2\ell \kappa&1&\kappa\\
 2\ell \kappa&\kappa&1\end{pmatrix},
 \label{eq:em-qutrit}
\end{equation}
with $\rho_2$ and $\rho_3$ obtained by cyclically shifting the basis
($1\to2\to3$). All three are positive definite for $\kappa_-<\kappa<\kappa_+$
(SM, Sec.~\ref{app:v3-exactness}), where
\begin{equation}
 \begin{gathered}
 \kappa_\pm=\frac{1\pm\sqrt{1+8\ell^2}}{4\ell^2}
 \\
 (\kappa_-\simeq-0.47397,\;\kappa_+\simeq0.90104).
 \end{gathered}
\end{equation}
Their average is $\tau=I/3+t(J-I)$, with
$t=\kappa(4\ell+1)/18$ and $J$ the all-ones matrix; $r=2$, and for
$\kappa\neq0$ the states do not commute pairwise. Nevertheless, for every
$\kappa$ the optimized posterior operators are
\begin{align}
 G_1^\star&=\operatorname{diag}(2,\tfrac12,\tfrac12),\quad
 G_2^\star=\operatorname{diag}(\tfrac12,2,\tfrac12),\nonumber\\
 G_3^\star&=\operatorname{diag}(\tfrac12,\tfrac12,2).
 \label{eq:em-qutrit-G}
\end{align}
They satisfy $\sum_ip_iG_i^\star=I$ and the only linear relation among the
$\rho_i-\tau$, so they are feasible. For diagonal $G$ the derivative of the
logarithm acts entrywise, $D\log G[X]_{kl}=X_{kl}/\operatorname{lm}(g_k,g_l)$,
with the logarithmic mean $\operatorname{lm}(a,b)=(a-b)/(\ln a-\ln b)$ and
$\operatorname{lm}(a,a)=a$. Every entry of $\rho_i$ equals $\operatorname{lm}$
of the corresponding eigenvalues of $G_i^\star$ times the entry of
$K=\frac13[(1-\kappa)I+\kappa J]$; for instance $\operatorname{lm}(2,\frac12)=\ell$.
Hence $D\log G_i^\star[\rho_i]=K$ for every $i$, and the stationarity condition
$\sum_ip_is_{ia}D\log G_i^\star[\rho_i]=0$
[SM, Sec.~\ref{app:v3-geometry}, Eq.~(\ref{eq:v3-sm-kkt})] holds because
$\sum_ip_is_{ia}=0$; strict concavity makes Eq.~\eqref{eq:em-qutrit-G} the
maximizer.
The computational basis gives posterior permutations of
$(\frac23,\frac16,\frac16)$ and
$\Iacc=\Uaff=(\ln2)/3$ nat $=1/3$ bit throughout the full-rank range.
For $\kappa\neq0$, $\chi>\Iacc$, with gap $0.0259$ nat at $\kappa=1/2$.

\clearpage
\onecolumngrid
\section*{Supplemental Material for\texorpdfstring{\\}{ }\papertitle}
\begin{center}
Jinbo Wang, Qihang Wang, Pengcheng Hou, and Kun Chen
\end{center}
\makeatletter
\let\@hangfrom@section\regularsectionheading
\let\@sectioncntformat\regularsectionnumber
\makeatother
\setcounter{section}{0}
\setcounter{equation}{0}
\setcounter{figure}{0}
\setcounter{table}{0}
\setcounter{theorem}{0}
\counterwithout{equation}{section}
\renewcommand{\thesection}{S\arabic{section}}
\renewcommand{\thesubsection}{\thesection.\Alph{subsection}}
\renewcommand{\theequation}{S\arabic{equation}}
\renewcommand{\thefigure}{S\arabic{figure}}
\renewcommand{\thetable}{S\Roman{table}}
\renewcommand{\thetheorem}{S\arabic{theorem}}
\renewcommand{\theHsection}{supp.\arabic{section}}
\renewcommand{\theHequation}{supp.\arabic{equation}}
\renewcommand{\theHfigure}{supp.\arabic{figure}}
\renewcommand{\theHtable}{supp.\arabic{table}}
\renewcommand{\theHtheorem}{supp.\arabic{theorem}}
\setcounter{secnumdepth}{1}
\renewcommand{\prlappendix}[1]{\section{#1}}
\newcommand{\SMsec}[1]{Sec.~\ref{#1}}

This Supplemental Material gives the complete proofs and numerical methods
behind the Letter. Sections, equations, figures, tables, and theorems carry
the prefix S; unprefixed equation and figure numbers, and End Matter A and
B, refer to the Letter. Sections~\ref{app:v3-geometry}--\ref{app:v3-scope}
develop the posterior algebra and its attained-bound classification,
prove Shor's conjecture and classify all binary optima, including singular
states, and establish the stability and receiver guarantees.
Sections~\ref{app:extension-overview}--\ref{app:pareto-physical} treat the
information--error trade-off, including the three-outcome construction, its
certificates, and physical examples. Section~\ref{app:pareto-general}
extends the analysis to unequal decision costs and other information scores.

The receiver benchmarks in Sec.~\ref{app:v3-benchmark} evaluate analytic
bounds in float64. Outward-rounded interval certificates for specified
information--error and Tsallis witnesses are given in
Secs.~\ref{app:pareto-size} and~\ref{app:tsallis-witness}; each section
separately states the precision and scope of its parameter scans.
The use of automated assistance in analytical checks and in reviewing the
numerical protocols and results is disclosed at the end of the Letter.

We use $\|\cdot\|_F$ for the Frobenius (Hilbert--Schmidt) norm and
$\|\cdot\|_\infty$ for the operator norm; TV denotes total-variation
distance. The relations $X\succ0$ and
$X\succeq0$ mean positive definiteness and positive semidefiniteness,
respectively. The measurement definitions and logarithm convention are
those of the Letter.

\prlappendix{Affine geometry and the information upper bound}
\label{app:v3-geometry}

Throughout this section, $p_i>0$ and every $\rho_i$ is positive definite.  Let
$\mathsf V=\operatorname{span}_{\mathbb R}\{\rho_i-\tau\}$ and choose a
minimal representation
\begin{equation}
 \rho_i=\tau+\sum_{a=1}^r s_{ia}A_a,
 \qquad \sum_i p_is_{ia}=0,
 \qquad G_i(\bm H)=I+\sum_as_{ia}H_a\succ0.
 \label{eq:v3-sm-affine}
\end{equation}
A change of the basis $\{A_a\}$ changes $\bm H$ but not the set of feasible
tuples $(G_i)_i$. These are exactly the tuples $G_i=I+X_i\succ0$ whose
Hermitian $X_i$ satisfy every real linear relation $\sum_ic_iX_i=0$ obeyed by the
$\rho_i-\tau$, and such relations do not refer to a basis. Since $\Phi$
depends on $\bm H$ only through the $G_i$, the value $\Uaff$, the posterior
operators $G_i^\star$, and the algebra they generate are basis independent.
The feasible domain is nonempty and convex.  Since
\(\sum_i p_iG_i(\bm H)=I\), every feasible posterior satisfies
\(G_i\preceq p_i^{-1}I\).  Full column rank of the coordinate matrix bounds
\(\bm H\).  If a feasible sequence approaches the boundary, some \(G_j\)
acquires an eigenvalue tending to zero.  Faithfulness of \(\rho_j\) implies
\(\Tr\rho_j\log G_j\to-\infty\), while all other terms remain bounded above.
The maximum is therefore attained in the interior.

For $X\succ0$, $R\succ0$, and Hermitian $K$, operator concavity of the
logarithm and its integral representation give
\begin{equation}
 \frac{\lambda_{\min}(R)}{\norm X_\infty^2}\norm K_F^2
 \leq-D^2\Tr R\log X[K,K]
 \leq\frac{\norm R_\infty}{\lambda_{\min}(X)^2}\norm K_F^2.
 \label{eq:v3-sm-hessian-pair}
\end{equation}
The Hessian follows by integrating
$2\Tr[R(X+uI)^{-1}K(X+uI)^{-1}K(X+uI)^{-1}]$.
The integrand is a squared Frobenius norm; bounding $R$ by its extreme
eigenvalues and then the resolvents by those of $X$ gives both constants.
For a tuple direction $\bm K$, put $K_i=\sum_as_{ia}K_a$.  Since
$G_i\preceq p_i^{-1}I$, Eq.~\eqref{eq:v3-sm-hessian-pair} yields
\begin{equation}
 -D^2\Phi(\bm H)[\bm K,\bm K]
 \geq\sum_i p_i^3\lambda_{\min}(\rho_i)\norm{K_i}_F^2.
 \label{eq:v3-sm-strong}
\end{equation}
Minimality of the coordinates makes the right side a norm, so \(\Phi\) is
strongly concave and its interior optimizer is unique.  Its stationarity
conditions are
\begin{align}
 \sum_i p_is_{ia}\mathcal L_{G_i^\star}(\rho_i)&=0,\nonumber\\
 \mathcal L_F(X)=D\log F[X]
 &=\int_0^\infty(F+uI)^{-1}X(F+uI)^{-1}du.
 \label{eq:v3-sm-kkt}
\end{align}

For a finite-outcome POVM without zero effects, set
\begin{align}
 c_y&=\Tr\tau M_y,
 &v_{ay}&=\frac{\Tr A_aM_y}{c_y},\nonumber\\
 H_{a,M}&=\sum_yv_{ay}M_y,
 &g_{iy}&=\frac{\Tr\rho_iM_y}{c_y}.
 \label{eq:v3-sm-embedding}
\end{align}
Then
\begin{equation}
 G_i(\bm H_M)=\sum_y g_{iy}M_y.
 \label{eq:v3-sm-induced}
\end{equation}
The map \(X\mapsto\log X\) is operator concave, hence
\begin{equation}
 \log G_i(\bm H_M)\succeq\sum_y(\log g_{iy})M_y.
 \label{eq:v3-sm-jensen}
\end{equation}
Tracing against \(p_i\rho_i\), summing over \(i\), and using the definition
of \(g_{iy}\) proves
\(I(\cE,M)\leq\Phi(\bm H_M)\leq\Uaff\).

For completeness, introduce the unrestricted value
\begin{equation}
 U_{\rm unres}=\max_{F_i\succ0,\ \sum_i p_iF_i=I}
 \sum_i p_i\Tr\rho_i\log F_i.
\end{equation}
Affine feasibility gives $\Uaff\leq U_{\rm unres}$.  For an unrestricted
feasible family, set $q_i=\Tr\tau F_i$ and $Q_i=F_i/q_i$.  Then
$\Tr\tau Q_i=1$, and the variational formula for measured relative entropy
implies
\begin{equation}
 \Tr\rho_i\log F_i
 \leq D_{\mathbb M}(\rho_i\Vert\tau)+\log q_i.
\end{equation}
Moreover, $\sum_ip_iq_i=1$, so scalar Jensen gives
$\sum_ip_i\log q_i\leq0$.  Consequently,
\begin{equation}
 \Iacc\leq\Uaff\leq U_{\rm unres}
 \leq\sum_i p_iD_{\mathbb M}(\rho_i\Vert\tau)\leq\chi(\cE).
 \label{eq:v3-sm-holevo}
\end{equation}
The variational formula is Lemma~1 and Theorem~2 of Ref.~\cite{berta2017variational}; see also the measured-divergence treatments
\cite{matsumoto2014maximization,hiai2021quantum}.
This last bound is not needed to prove the exactness criterion.
It shows that \(\Uaff\) is an upper bound on accessible information
that is no larger than the Holevo quantity.

\prlappendix{Attainment of the bound and all optimal measurements}
\label{app:v3-exactness}

Under the full-rank assumptions of Sec.~\ref{app:v3-geometry}, we prove
the attained-bound criterion and complete refinement classification of
Eqs.~\eqref{eq:letter-commutation} and~\eqref{eq:v3-refinement}.
The argument establishes finite-outcome attainment, identifies equality in
compression Jensen, and then classifies the optimizers.
Let $d=\dim\mathcal H$. We first show directly that a finite-outcome POVM attains the maximum information.
Refine every effect into rank-one effects and write each as $d w P$, where
$P$ is a rank-one projector and the nonnegative weights sum to one.
Completeness is the moment constraint $\int P\,\mathrm d\mu(P)=I/d$.
The information is the continuous linear functional
\[
 d\int\!\left[\sum_i p_i\Tr(\rho_iP)
 \log\frac{\Tr(\rho_iP)}{\Tr(\tau P)}\right]\mathrm d\mu(P).
\]
The rank-one projector space is compact, and the feasible probability measures
form a compact convex set. Hence a maximum is attained at an extreme measure.
Such a measure has at most $d^2$ atoms: more disjoint positive-mass support sets
would give a linear dependence among their Hermitian matrix moments, and a
bounded signed density perturbation would preserve the moments in both
directions. The trace constraint also preserves total mass, contradicting
extremality. Refinement cannot reduce mutual information, so this supplies a
finite optimizer without an additional outcome-count assumption.

Next, we record equality in compression Jensen.  Let
\(V:\cH\to\mathcal K\) be an isometry, \(T\succ0\), and \(G=V^*TV\).
In the decomposition
\(\operatorname{ran}V\oplus(\operatorname{ran}V)^\perp\), write
\begin{equation}
 T=\begin{pmatrix}G&B\\B^*&N\end{pmatrix}.
\end{equation}
The integral representation of the logarithm and the Schur complement yield
\begin{align}
 \log G-V^*(\log T)V
 =\int_0^\infty\!\big[&
 (G+uI-B(N+uI)^{-1}B^*)^{-1}\nonumber\\
 &-(G+uI)^{-1}\big]du\succeq0.
 \label{eq:v3-sm-compression}
\end{align}
Equality holds if and only if \(B=0\), equivalently \(TV=VG\).  For the
POVM dilation
\begin{equation}
 Vx=\bigoplus_yM_y^{1/2}x,
 \qquad T_i=\bigoplus_yg_{iy}I,
 \label{eq:v3-sm-dilation}
\end{equation}
this gives
\begin{equation}
 g_{iy}M_y^{1/2}=M_y^{1/2}G_i
 \quad\text{for every }i,y.
 \label{eq:v3-sm-intertwining}
\end{equation}

Suppose first that the optimized posteriors commute.  Let
\(P_\star=\{P_z\}_z\) be their minimal joint-spectral PVM and write
\(G_i^\star P_z=g_{iz}P_z\).  Define
\begin{equation}
 \alpha_{iz}=\frac{p_i\Tr(\rho_iP_z)}{\Tr(\tau P_z)},
 \qquad \beta_{iz}=p_ig_{iz}.
\end{equation}
Both are probability vectors for each \(z\), and a direct expansion gives
\begin{equation}
 I(\cE,P_\star)-\Phi(\bm H_\star)
 =\sum_z\Tr(\tau P_z)D(\alpha_{\cdot z}\Vert\beta_{\cdot z})\geq0.
 \label{eq:v3-sm-gap}
\end{equation}
Since \(I(\cE,P_\star)\leq\Iacc\leq\Uaff=\Phi(\bm H_\star)\), equality
holds throughout and \(P_\star\) is optimal.

Conversely, assume \(\Iacc=\Uaff\).  Finite-dimensional attainment supplies
an optimal finite-outcome POVM \(M\).  Equality in
\begin{equation}
 I(\cE,M)\leq\Phi(\bm H_M)\leq\Phi(\bm H_\star)
\end{equation}
forces \(\bm H_M=\bm H_\star\) by uniqueness.  The positive difference in Jensen's inequality has zero trace against
each positive-definite state, so the difference itself is the zero operator.  Equation~\eqref{eq:v3-sm-intertwining} holds with
\(G_i=G_i^\star\).  The support of every nonzero effect lies in a common
eigenspace of all \(G_i^\star\); POVM completeness makes these supports span
\(\cH\).  Hence the optimized posterior operators commute.

It remains to classify the optimizers. Write $s_i=(s_{i1},\ldots,s_{ir})^{\mathsf T}$
for the coordinate columns and $B_z=(B_{1z},\ldots,B_{rz})^{\mathsf T}$
for the column of compressed operators below. Compress
Eq.~\eqref{eq:v3-sm-kkt} to a joint spectral block \(P_z\) and put
\begin{equation}
 T_z=P_z\tau P_z,
 \qquad B_{az}=P_zA_aP_z,
 \qquad C_z=\sum_i\frac{p_i}{g_{iz}}s_is_i^{\mathsf T}.
\end{equation}
Minimal affine coordinates make each $H_a^\star$ a real linear
combination of the $G_i^\star-I$. It is therefore scalar on every minimal
joint-spectral block, which defines the coefficients $h_{az}$ below.
The matrix \(C_z\) is positive definite.  If
\(h_z=(h_{1z},\ldots,h_{rz})^{\mathsf T}\), normalization gives
\(\sum_i(p_i/g_{iz})s_i+C_zh_z=0\), while the compressed stationarity
equations give \(C_z(B_z-h_zT_z)=0\).  Therefore
\begin{equation}
 B_{az}=h_{az}T_z,
 \qquad P_z\rho_iP_z=g_{iz}P_z\tau P_z.
 \label{eq:v3-sm-proportional}
\end{equation}
Equation~\eqref{eq:v3-sm-intertwining} shows that every optimal effect is
supported in one block.  Conversely, any block refinement satisfies
\begin{equation}
 \Tr\rho_iM_y=g_{iz}\Tr\tau M_y,
 \qquad y\in Y_z.
\end{equation}
Thus the conditional refinement channel
\begin{equation}
 W(y|z)=\frac{\Tr(\tau M_y)}{\Tr(\tau P_z)}
\end{equation}
is independent of the label \(i\), while \(y\mapsto z\) is deterministic.
The induced classical experiments are Blackwell-equivalent
\cite{blackwell1953equivalent,kuramochi2018experiments}.

\prlsubheading{Positivity of the qutrit example}
For the family in End Matter B, Eq.~\eqref{eq:em-qutrit}, put
$\ell=3/(4\ln2)$. The antisymmetric eigenvector has eigenvalue
$(1-\kappa)/6$, and the remaining $2\times2$ block is positive precisely
when $1+\kappa-2\ell^2\kappa^2>0$. Its roots are the $\kappa_\pm$ given
there, with $\kappa_+<1$; hence all three cyclically related states are
positive definite for $\kappa_-<\kappa<\kappa_+$.

The strict Holevo gap in Fig.~\ref{fig:posterior-algebra} follows from
\[
 [\rho_1,\tau]_{12}
 =\frac{\kappa(4\ell+1)[3+(2\ell-1)\kappa]}{108}.
\]
Here $1<\ell<3/2$ and $\kappa_->-1$, so the bracket is positive throughout
the full-rank interval. For $\kappa\ne0$, $\rho_1$ and $\tau$ therefore
do not commute. Proposition~5 of Ref.~\cite{berta2017variational} makes
$D_{\mathbb M}(\rho_1\Vert\tau)<D(\rho_1\Vert\tau)$, so the positive
prior and Eq.~\eqref{eq:v3-sm-holevo} give $\Iacc<\chi$.

\prlappendix{Two states: construction and singular states}
\label{app:v3-binary}

If one prior vanishes, or if the two states coincide, every measurement has
zero mutual information and any rank-one PVM is optimal.  We henceforth take
priors $p$ and $q=1-p$ with $0<p<1$, and distinct states; in this section
$q$ always denotes $1-p$, and outcome probabilities are written $c_y$.

\prlsubheading{Full-rank construction}
For a full-rank binary ensemble
\(\cE=\{(p,\rho_0),(q,\rho_1)\}\), define
\begin{align}
 \Phi_{\rm bin}(H)=&p\Tr\rho_0\log(I+qH)\nonumber\\
 &+q\Tr\rho_1\log(I-pH),
 \label{eq:v3-sm-binary-program}
\end{align}
with \(I+qH\succ0\) and \(I-pH\succ0\), and choose
\begin{equation}
 \rho_0=\tau+q(\rho_0-\rho_1),
 \qquad
 \rho_1=\tau-p(\rho_0-\rho_1).
\end{equation}
The affine program is the single-variable maximization of
Eq.~\eqref{eq:v3-sm-binary-program}.  Its optimized posteriors are
\begin{equation}
 G_0^\star=I+qH_\star,
 \qquad
 G_1^\star=I-pH_\star,
\end{equation}
so they are functions of one generator and commute.  Their common spectral
PVM is the spectral PVM of \(H_\star\), and the exactness theorem
proves that it attains \(\Iacc\).  Refining degenerate spectral projectors to
rank one cannot decrease mutual information and therefore preserves
optimality.

\prlsubheading{Singular states and complete classification}
We now allow singular signal states and classify all optimal finite-outcome POVMs.  Restrict to
$\cH_0=P_0\cH$, where $P_0=\supp\tau$.  Since $p,q>0$,
$\supp\rho_i\subseteq\cH_0$, and $\tau$ is full rank on $\cH_0$.  On the compact
operator interval
\begin{equation}
 \overline{\mathcal D}=\{H:-q^{-1}I\preceq H\preceq p^{-1}I\},
 \label{eq:v3-sm-closed-domain}
\end{equation}
set $G_0(H)=I+qH$ and $G_1(H)=I-pH$.  For $G\succeq0$, use the extended
convention
\begin{equation}
 \Tr\rho\log G=
 \begin{cases}
  \Tr_{\supp G}\rho\log(G|_{\supp G}),&\supp\rho\subseteq\supp G,\\
  -\infty,&\text{otherwise}.
 \end{cases}
 \label{eq:v3-sm-extended-log}
\end{equation}
The resulting $\overline\Phi_{\rm bin}$ is concave and upper semicontinuous,
so it attains its maximum on $\overline{\mathcal D}$.  The value at $H=0$ is
finite.

The maximizer is unique.  Indeed, suppose $H\ne K$ were maximizers and examine
an interior point of their segment.  Common kernel blocks of $G_i(H)$ and $G_i(K)$ carry no weight in
$\rho_i$, by Eq.~\eqref{eq:v3-sm-extended-log};
positivity also forces the segment direction to have no row or column joining
such a block.  On the remaining support, the integral representation of the
logarithm gives, for $D=H-K$,
\begin{equation}
 -\frac{d^2}{dt^2}\Tr\rho_i\log G_i(K+tD)
 =2s_i^2\!\int_0^\infty\!\Tr\rho_i R_{i,u}D R_{i,u}D R_{i,u}\,du\geq0,
 \label{eq:v3-sm-closed-strict}
\end{equation}
where $s_0=q$, $s_1=-p$, and $R_{i,u}=(G_i+uI)^{-1}$.
The integrand equals
$\norm{R_{i,u}^{1/2}D R_{i,u}\rho_i^{1/2}}_F^2$.
Equality for all $u$
implies $D\rho_i^{1/2}=0$ by taking $u\to\infty$.  Equality in the concavity
of the sum would therefore make $D$ annihilate both signal supports.  Their
ranges span $\supp\tau=\cH_0$, so $D=0$, a contradiction.  Denote the unique
maximizer by $H_\star$.

Every finite-outcome POVM on $\cH_0$ defines a feasible point of the closed
optimization problem.  After deleting
zero effects, put
\begin{equation}
 c_y=\Tr\tau M_y>0,\quad
 h_y=\frac{\Tr(\rho_0-\rho_1)M_y}{c_y},\quad
 g_{0y}=1+qh_y,\quad g_{1y}=1-ph_y,
 \label{eq:v3-sm-closed-embedding}
\end{equation}
and $H_M=\sum_yh_yM_y$.  Then
$G_i(H_M)=\sum_yg_{iy}M_y\succeq0$.  Applying compression Jensen to
$T_{i,\epsilon}=\bigoplus_y(g_{iy}+\epsilon)I$ and taking
$\epsilon\downarrow0$ yields
\begin{equation}
 I(\cE,M)\leq\overline\Phi_{\rm bin}(H_M)
 \leq\overline\Phi_{\rm bin}(H_\star).
 \label{eq:v3-sm-closed-upper}
\end{equation}
Terms with $g_{iy}=0$ contribute zero because then
$\Tr\rho_iM_y=c_yg_{iy}=0$.  These identities also show
$\supp\rho_i\subseteq\supp G_i(H_M)$: if $x\in\ker G_i(H_M)$, positivity of
the summands gives $M_yx=0$ when $g_{iy}>0$, whereas $g_{iy}=0$ implies
$\rho_i^{1/2}M_y=0$; summing over $y$ yields $\rho_i^{1/2}x=0$.

Let $P_\star=\{P_z\}_z$ be the spectral PVM of $H_\star$, with eigenvalues
$h_z$ and $g_{0z}=1+qh_z$, $g_{1z}=1-ph_z$.  Finiteness of the objective
implies $P_z\rho_iP_z=0$ whenever $g_{iz}=0$.  Thus the gap identity
\begin{equation}
 I(\cE,P_\star)-\overline\Phi_{\rm bin}(H_\star)
 =\sum_z\Tr(\tau P_z)D(\alpha_{\cdot z}\Vert\beta_{\cdot z})\geq0
 \label{eq:v3-sm-closed-gap}
\end{equation}
remains valid with $\alpha_{iz}=p_i\Tr(\rho_iP_z)/\Tr(\tau P_z)$,
$\beta_{iz}=p_ig_{iz}$, and the usual zero conventions.  Combining
Eqs.~\eqref{eq:v3-sm-closed-upper} and \eqref{eq:v3-sm-closed-gap} shows
\begin{equation}
 I(\cE,P_\star)=\Iacc=\max_{H\in\overline{\mathcal D}}
 \overline\Phi_{\rm bin}(H).
 \label{eq:v3-sm-closed-exact}
\end{equation}

We next classify all optimal finite-outcome POVMs.
We use the following boundary form of compression-Jensen equality.  Let
$T\succeq0$, $G=V^*TV$, and $\rho\succeq0$, with the extended logarithmic
traces finite.  If the limiting weighted Jensen gap vanishes, then
\begin{equation}
 (TV-VG)\rho^{1/2}=0.
 \label{eq:v3-sm-weighted-boundary-jensen}
\end{equation}
To see this, replace $T$ by $T+\epsilon I$ and use the block decomposition in
Eq.~\eqref{eq:v3-sm-compression}.  If $B$ is its off-diagonal block, then $B$
is independent of $\epsilon$, and the corresponding positive resolvent
integrand $K_\epsilon(u)$ satisfies, uniformly for $0\leq\epsilon\leq1$,
\begin{equation}
 u^3K_\epsilon(u)\longrightarrow BB^*\qquad(u\to\infty).
\end{equation}
Consequently, if $B^*\rho^{1/2}\ne0$, then on some fixed interval
$u\in[R,2R]$ the weighted integrand has a strictly positive lower bound,
independent of sufficiently small $\epsilon$.  This contradicts a vanishing
limiting Jensen gap.  Hence $B^*\rho^{1/2}=0$, which is precisely
Eq.~\eqref{eq:v3-sm-weighted-boundary-jensen}.

If $M$ is optimal, both inequalities in Eq.~\eqref{eq:v3-sm-closed-upper}
are equalities and uniqueness gives $H_M=H_\star$.  Applying
Eq.~\eqref{eq:v3-sm-weighted-boundary-jensen} to each signal and outcome gives
\begin{equation}
 M_y^{1/2}(H_\star-h_yI)\rho_0^{1/2}=0,\qquad
 M_y^{1/2}(H_\star-h_yI)\rho_1^{1/2}=0.
 \label{eq:v3-sm-singular-annihilation}
\end{equation}
The two signal ranges span $\cH_0$, hence
$M_y^{1/2}(H_\star-h_yI)=0$.  Consequently every nonzero $M_y$ is supported
in one spectral block $P_z$.

Finally, the signal states are proportional after compression to each block:
\begin{equation}
 P_z\rho_iP_z=g_{iz}P_z\tau P_z.
 \label{eq:v3-sm-singular-proportional}
\end{equation}
For $-q^{-1}<h_z<p^{-1}$ this follows by varying $H$ arbitrarily inside
$P_z$ and using stationarity; since $G_i^\star|_{P_z}=g_{iz}P_z$, the resulting
condition is $P_z\rho_0P_z/g_{0z}=P_z\rho_1P_z/g_{1z}$.  At an endpoint,
one $g_{iz}$ vanishes, and Eq.~\eqref{eq:v3-sm-extended-log} gives the same
identity directly.  Therefore the optimal finite-outcome POVMs on $\cH_0$ are exactly
the block refinements
\begin{equation}
 M_y=P_zM_yP_z\ (y\in Y_z),\qquad
 \sum_{y\in Y_z}M_y=P_z.
 \label{eq:v3-sm-singular-classification}
\end{equation}
Their refinement channel does not depend on the input label, so every one
is Blackwell-equivalent to $P_\star$.  Rank-one refinement proves Shor's
orthogonal-measurement statement.  On the full Hilbert space, every POVM
has the same signal statistics as its compression $P_0M_yP_0$: the part of
a POVM acting outside $\supp\tau$ never registers a signal and does not
affect the statistics.

\prlsubheading{Relation to existing variational and projective-reduction results}
The binary variational formulation of Jo~\cite{jo2026shor} is related to
Eq.~\eqref{eq:v3-sm-binary-program} by $T=p(I+qH)$, after matching signal
labels. Indeed, $I-T=q(I-pH)$, and for an embedded measurement
$T_M=\sum_y\Pr(0|y)M_y=p(I+qH_M)$. These affine changes preserve the
spectral projectors. Applying the gap identity to $H_M$, rather than only
to $H_\star$, gives
$I(\cE,M)\leq\overline\Phi_{\rm bin}(H_M)\leq I(\cE,P_{H_M})$, which is
the pointwise spectral replacement of Theorem~3.1 of Ref.~\cite{jo2026shor}.
The variational values therefore agree. The equality analysis above also
classifies every finite-outcome optimizer, including for singular states.

For comparison, we derive projective sufficiency from the theorem of
Fang, Fawzi, and Fawzi.  Their
finite-dimensional projective-reduction theorem states the following.  Let
$f$ be convex and lower semicontinuous with
$(0,\infty)\subset\operatorname{dom}f$.  If an interval $J$ admits a
one-to-one parametrization $\psi:J\to\operatorname{dom}f^*$ for which $\psi$
is operator concave and $f^*\!\circ\psi$ is operator convex, then the measured
$f$-divergence optimized over POVMs equals the value optimized over PVMs
\cite[Theorem~2]{fang2026uhlmann}.

If all signals in a collinear ensemble coincide, every measurement has zero
mutual information.  Otherwise choose two extreme signal states $\rho_0$ and
$\rho_1$ and write
$\rho_i=(1-\lambda_i)\rho_0+\lambda_i\rho_1$, where
$0\leq\lambda_i\leq1$, some $\lambda_i=0$, some $\lambda_i=1$, and
$\bar\lambda=\sum_ip_i\lambda_i\in(0,1)$.  The mutual information of a
measurement is the measured $f$-divergence of the two endpoint outcome
distributions for
\begin{equation}
 f(t)=\sum_i p_i[(1-\lambda_i)+\lambda_it]
 \log\frac{(1-\lambda_i)+\lambda_it}
 {(1-\bar\lambda)+\bar\lambda t}.
 \label{eq:v3-sm-f}
\end{equation}
Here the displayed formula is for $t>0$; set $f(0)=\lim_{t\downarrow0}f(t)$
and $f(t)=+\infty$ for $t<0$ when taking the convex conjugate.
Set
\begin{equation}
 h=\frac{t-1}{(1-\bar\lambda)+\bar\lambda t},
 \qquad
 \psi(h)=\sum_i p_i\lambda_i
 \log[1+(\lambda_i-\bar\lambda)h].
\end{equation}
On $J=(-1/(1-\bar\lambda),1/\bar\lambda)$, direct differentiation and
Fenchel conjugacy give
\begin{align}
 f'(t)&=\psi(h),\\
 f^*(\psi(h))
 &=-\sum_i p_i(1-\lambda_i)
 \log[1+(\lambda_i-\bar\lambda)h].
 \label{eq:v3-sm-conjugate}
\end{align}
The first expression is operator concave and the second is operator convex.
Furthermore,
\begin{equation}
 \psi'(h)=(1-\bar\lambda h)
 \sum_i\frac{p_i(\lambda_i-\bar\lambda)^2}
 {1+(\lambda_i-\bar\lambda)h}>0.
\end{equation}
The M\"obius map sends $t\in(0,\infty)$ bijectively onto $J$.  Since
$\psi(h)=f'(t)$ and $\psi'>0$, one-dimensional Fenchel duality identifies its
range with $\operatorname{dom}f^*$.  Indeed, the signal with $\lambda_i=1$
gives $f'(0+)=-\infty$, while the signal with $\lambda_i=0$ makes the opposite
limiting slope inadmissible because
$f(t)-t f'(\infty)=-w_0\log t+O(1)\to-\infty$, where
$w_0=\sum_{i:\lambda_i=0}p_i>0$. Thus the opposite endpoint has infinite
conjugate value and no point of the conjugate domain is omitted.  With the
lower-semicontinuous conventions at ratios zero and infinity, all hypotheses
of Fang--Fawzi--Fawzi's theorem hold.
Therefore the POVM and PVM suprema agree; rank-one refinement and compactness
give attainment.  Thus projective sufficiency follows from that earlier theorem. The present
construction also identifies the canonical optimized posterior operators, classifies
all finite optimal measurements, and gives the stability and finite-iterate
bounds of Sec.~\ref{app:v3-robust}.

\prlappendix{Stability of nearly optimal measurements}
\label{app:v3-robust}

\textit{Near-optimal measurements}---How close must a nearly optimal
measurement stay to these eigenspaces? For full-rank binary signals, a
measurement that loses $\epsilon_M=\Iacc-I(\mathcal E,M)$ of information
cannot put much effect mass outside the eigenspaces $P_z$ of $H_\star$:
the total trace weight $L(M)$ that its effects place outside their dominant
eigenspaces obeys $L(M)\leq8\epsilon_M/\Delta_{\rm post}^2$. Here
$\Delta_{\rm post}^2=\min_{z\ne z'}\sum_i a_i(g_{iz}-g_{iz'})^2$ measures
how distinct the posterior ratios of different eigenspaces are, with weights
$a_i=p_i^3\lambda_{\min}(\rho_i)$. The bound weakens
as the states approach singularity or eigenspaces merge.

We next bound the deviation from the equality condition for a measurement
whose information is close to optimal.  If
\(0\prec T\preceq\beta I\), the block representation used in
Eq.~\eqref{eq:v3-sm-compression} and the resolvent identity imply
\begin{equation}
 \Tr[\log G-V^*(\log T)V]
 \geq\frac{1}{2\beta^2}\norm{TV-VG}_F^2.
 \label{eq:v3-sm-compression-remainder}
\end{equation}
Indeed, with \(G_u=G+uI\) and \(N_u=N+uI\), the integrand of
Eq.~\eqref{eq:v3-sm-compression} is at least
\(G_u^{-1}BN_u^{-1}B^*G_u^{-1}\), and its trace is bounded below by
\((\beta+u)^{-3}\norm B_F^2\).  Integration gives the stated constant and
\(\norm B_F=\norm{TV-VG}_F\).  If \(\rho\succeq\mu I\), the same positive difference gives the
corresponding lower bound after tracing against \(\rho\).

Define the weighted posterior norm
\begin{equation}
 \norm{\bm K}_{\rm post}^2
 =\sum_i a_i\norm{K_i}_F^2,
 \qquad a_i=p_i^3\lambda_{\min}(\rho_i).
 \label{eq:v3-sm-post-norm}
\end{equation}
Equation~\eqref{eq:v3-sm-strong} says that \(\Phi\) is one-strongly
concave in this norm.  For a general full-rank ensemble, measure the
loss relative to the bound, $\epsilon_M=\Uaff-I(\cE,M)$; this agrees with
the definition above whenever $\Iacc=\Uaff$, in particular for two
full-rank states.  Split this gap as
\begin{equation}
 \epsilon_M=\Uaff-I(\cE,M)
 =\epsilon_{\rm opt}+\epsilon_J,
 \quad
 \begin{cases}
 \epsilon_{\rm opt}=\Phi(\bm H_\star)-\Phi(\bm H_M),\\
 \epsilon_J=\Phi(\bm H_M)-I(\cE,M).
 \end{cases}
\end{equation}
Strong concavity gives
$\sum_i a_i\norm{G_i(\bm H_M)-G_i^\star}_F^2\leq2\epsilon_{\rm opt}$.
For the POVM dilation, $T_i\preceq p_i^{-1}I$ because
$g_{iy}=q(y|i)/q(y)\leq p_i^{-1}$.  Applying
Eq.~\eqref{eq:v3-sm-compression-remainder} and then tracing against
$p_i\rho_i$ gives the corresponding residual relative to
$G_i(\bm H_M)$ bounded by $2\epsilon_J$.  Minkowski's inequality in the
weighted direct sum then yields
\begin{equation}
 R_\star(M)=\sum_{i,y}a_i
 \norm{g_{iy}M_y^{1/2}-M_y^{1/2}G_i^\star}_F^2
 \leq4\epsilon_M.
 \label{eq:v3-sm-universal}
\end{equation}
No commutativity assumption enters this estimate.

If there is only one joint spectral block, every POVM already refines it
and $L(M)=0$. For the separation bounds below, assume at least two distinct
joint spectral blocks. When the optimized posterior operators commute, write
\(G_i^\star=\sum_zg_{iz}P_z\) and define
\begin{equation}
 \Delta_{\rm post}^2=
 \min_{z\ne z'}\sum_i a_i(g_{iz}-g_{iz'})^2.
 \label{eq:v3-sm-separation}
\end{equation}
Write $g_z=(g_{1z},\ldots,g_{mz})$, $g_y=(g_{1y},\ldots,g_{my})$, and
$\norm v_a^2=\sum_i a_iv_i^2$.  For an outcome \(y\), expansion over the
spectral blocks gives
\begin{equation}
 \sum_i a_i\norm{M_y^{1/2}G_i^\star-g_{iy}M_y^{1/2}}_F^2
 =\sum_z\norm{g_z-g_y}_a^2\Tr P_zM_y.
\label{eq:v3-sm-residual-identity}
\end{equation}

\begin{lemma}[Weighted separation]\label{lem:v3-sm-variance}
Let $u_1,\ldots,u_K$ obey $\norm{u_z-u_{z'}}_a^2\geq\Delta^2$ for
$z\ne z'$, and let $w_z\geq0$, $W=\sum_zw_z$.  For every $u$,
\begin{equation}
 \sum_zw_z\norm{u_z-u}_a^2
 \geq\frac{\Delta^2}{2}\bigl(W-\max_zw_z\bigr).
 \label{eq:v3-sm-variance}
\end{equation}
\end{lemma}

\par\textit{Proof}---
If $W=0$, both sides vanish. For $W>0$, the left side is minimized at $\bar u=W^{-1}\sum_zw_zu_z$.  The weighted
variance identity gives
\begin{equation}
 \sum_zw_z\norm{u_z-\bar u}_a^2
 =\frac{1}{2W}\sum_{z,z'}w_zw_{z'}\norm{u_z-u_{z'}}_a^2.
\end{equation}
The off-diagonal terms are at least $\Delta^2$, and
$\sum_zw_z^2\leq W\max_zw_z$, proving
Eq.~\eqref{eq:v3-sm-variance}.
\hfill\ensuremath{\blacksquare}\par

Assign \(y\) to a block $z(y)$ maximizing \(\Tr P_zM_y\).  Apply
Lemma~\ref{lem:v3-sm-variance} with $w_z=\Tr P_zM_y$, sum over outcomes,
and use Eq.~\eqref{eq:v3-sm-universal}.  This proves the wrong-block mass
bound
\begin{equation}
 L(M)=\sum_y\sum_{z\ne z(y)}\Tr P_zM_y
 \leq\frac{8\epsilon_M}{\Delta_{\rm post}^2}.
 \label{eq:v3-sm-leakage}
\end{equation}
The explicit effects
\begin{equation}
 B_y=P_{z(y)}M_yP_{z(y)},
 \qquad
 D_z=P_z-\sum_{y:z(y)=z}B_y
\end{equation}
form an exact block refinement $\widetilde M=\{B_y\}_y\cup\{D_z\}_z$,
with zero effects omitted. For experiments on the fixed signal ensemble,
define the directed deficiency from $N$ to $M$ as
$\inf_K\max_i\operatorname{TV}(q_M(\cdot|i),Kq_N(\cdot|i))$,
where $K$ ranges over stochastic maps; $\delta_{\rm LR}$ is the maximum
of the two directed deficiencies. Trace-norm Cauchy--Schwarz gives a symmetric
finite-experiment deficiency bound
\begin{equation}
 \delta_{\rm LR}(M,\widetilde M)
 \leq\min\{1,L(M)+\sqrt{L(M)[d-L(M)]}\}.
\label{eq:v3-sm-deficiency}
\end{equation}
For completeness, put $P=P_{z(y)}$, $Q=I-P$, and
$L_y=\Tr QM_yQ$.  Then
\begin{equation}
 \norm{M_y-B_y}_1
 \leq L_y+2\sqrt{(\Tr PM_yP)L_y}.
 \label{eq:v3-sm-effect-error}
\end{equation}
Summing and using Cauchy--Schwarz once more gives
\begin{equation}
 \sum_y\norm{M_y-B_y}_1
 \leq L(M)+2\sqrt{L(M)[d-L(M)]}.
\end{equation}
To simulate $M$ from $\widetilde M$, report $y$ after outcome $B_y$ and
distribute each $D_z$ according to any fixed probability distribution over
all original outcome labels. This remains defined even if no $y$ was
assigned to block $z$. In the reverse
direction, map outcome $y$ of $M$ to the correspondingly named outcome of
$\widetilde M$ and never report a $D_z$ label.  For every signal state, the
two resulting total-variation errors are bounded by one half of the last
display plus $L(M)/2$.  Taking the larger directed deficiency proves
Eq.~\eqref{eq:v3-sm-deficiency} (and the trivial upper bound is one).
Thus the outcome distributions of a nearly optimal measurement can be
approximated by those of a refinement of the canonical posterior PVM,
after classical postprocessing in either direction. The bound concerns
these distributions for the fixed signal ensemble; it does not give
operator-by-operator or diamond-norm closeness of the measurements.

At a feasible numerical tuple \(\widehat{\bm H}\), let
\(R_a=\nabla_a\Phi(\widehat{\bm H})\) and
\begin{equation}
 C=\sum_i a_is_is_i^{\mathsf T},
 \qquad
 r_{\rm g}^2=\sum_{a,b}(C^{-1})_{ab}\Tr(R_aR_b).
\end{equation}
This is the squared dual norm of the gradient for the weighted norm
defined in Eq.~\eqref{eq:v3-sm-post-norm}.  One-strong
concavity implies, for $\bm K=\bm H_\star-\widehat{\bm H}$,
\begin{equation}
 \Phi(\bm H_\star)-\Phi(\widehat{\bm H})
 \leq\langle\nabla\Phi(\widehat{\bm H}),\bm K\rangle
 -\tfrac12\norm{\bm K}_{\rm post}^2.
\end{equation}
Dual Cauchy--Schwarz and maximization of the scalar quadratic give the first
bound below. For the second, strong monotonicity of $-\nabla\Phi$ and
stationarity at $\bm H_\star$ imply
$\norm{\bm K}_{\rm post}^2\leq
\langle\nabla\Phi(\widehat{\bm H}),\bm K\rangle
\leq r_{\rm g}\norm{\bm K}_{\rm post}$:
\begin{equation}
 \Uaff\leq\Phi(\widehat{\bm H})+\tfrac12r_{\rm g}^2,
 \qquad
 \sum_i a_i\norm{G_i(\widehat{\bm H})-G_i^\star}_F^2\leq r_{\rm g}^2.
 \label{eq:v3-sm-finite-iterate}
\end{equation}
The first inequality is the binary stopping certificate used in Sec.~\ref{app:v3-benchmark}.
A formally computer-assisted certificate additionally requires outward-rounded
matrix functions, eigenvalue bounds, and mutual-information evaluation; the
reported benchmark uses float64 arithmetic.

\prlappendix{Posterior-spectral receiver and numerical benchmarks}
\label{app:v3-benchmark}

The single generator gives a practical receiver. For full-rank signals one
optimizes one Hermitian matrix rather than all POVM effects, and every trial
matrix carries a two-sided bound. Diagonalizing a feasible $H$ gives a PVM
$P_H$ and
\begin{equation}
 I_-=I(\mathcal E,P_H)\leq\Iacc\leq I_+=\Phi(H)+\frac{\Tr[R(H)^2]}{2\omega},
 \label{eq:v3-certificate}
\end{equation}
where $R(H)=\nabla\Phi(H)$ is the gradient (given below) and
$\omega=p^2q^2[p\lambda_{\min}(\rho_0)+q\lambda_{\min}(\rho_1)]$, with
$q=1-p$ and $\lambda_{\min}$ the smallest eigenvalue, is a lower bound on the
curvature of $-\Phi$. The lower bound is the information delivered by $P_H$; the upper bound
follows from strong concavity. Thus $I_+-I_-$ bounds how far $P_H$ is from optimal
in exact arithmetic, regardless of how $H$ was found
(for certificates in convex quantum optimization, see
Ref.~\cite{coutts2021certifying}). The constant $\omega$ is conservative: in
numerical tests the actual curvature exceeded it by a factor of at least
$1.5$ (median about $14$ at random interior $H$, and far larger near the
boundary of the domain), so the analytic gap can overstate the true distance
from optimal. The benchmark evaluates this bound in float64;
an interval certificate additionally requires outward-rounded matrix functions,
eigenvalue bounds, and information values. Algorithm~\ref{alg:v3-binary} below uses
positivity-preserving backtracking ascent; the benchmark uses inverse-BFGS quasi-Newton
proposals.

On 408 stress cases, the receiver achieves a float64 upper bound on
information loss of at most $10^{-6}$ nat in 395 cases ($96.8\%$) within
300 iterations (Fig.~\ref{fig:full-solver}(b)). Against the same upper
bound, multistart PVM optimization and direct POVM ascent reach this target
in 169 and 85 cases ($41.4\%$ and $20.8\%$). Of the 13 receiver misses,
eight in dimension 16 and one with spectral floor $10^{-6}$ reach the
iteration cap; four with that floor stop earlier, after 121--142
iterations, on a gradient-tolerance test of the implementation that
Algorithm~1 omits. These are fixed-budget
comparisons of the implementations below, not complexity bounds. In
particular, a wide upper bound does not imply a correspondingly large
actual information loss. Readout requires a basis rotation and an
orthogonal measurement, at device-dependent cost and accuracy.

For the binary affine coordinates $s_0=q$, $s_1=-p$, the dual gradient
norm of Sec.~\ref{app:v3-robust} satisfies
$r_{\rm g}^2=\Tr[R(H)^2]/\omega$. Thus its general finite-iterate bound,
Eq.~\eqref{eq:v3-sm-finite-iterate}, gives Eq.~\eqref{eq:v3-certificate}.
Algorithm~\ref{alg:v3-binary} evaluates this interval after each update.

\par\medskip
\noindent\begin{minipage}{\columnwidth}
\refstepcounter{algorithm}\label{alg:v3-binary}
\hrule\smallskip
\textbf{Algorithm \thealgorithm: Posterior-spectral receiver}
\smallskip\hrule\smallskip
\textit{Input:} full-rank states $\rho_0,\rho_1$; prior $0<p<1$;
target gap $\varepsilon>0$; iteration limit $N$; initial step
$\alpha_0>0$ and backtracking parameters $0<\beta,\eta<1$.

\smallskip
\textit{Output:} a rank-one PVM $P$, information interval $[I_-,I_+]$,
and a flag indicating whether $I_+-I_-\leq\varepsilon$.
\smallskip
\begin{list}{\arabic{enumi}.}{%
 \usecounter{enumi}\setlength{\leftmargin}{1.4em}%
 \setlength{\labelsep}{0.35em}\setlength{\itemsep}{3pt}%
 \setlength{\parsep}{0pt}\setlength{\topsep}{0pt}}
\item Set $q\leftarrow1-p$, $H\leftarrow0$, and compute the curvature
constant $\omega$ defined below Eq.~\eqref{eq:v3-certificate}.
\item \textbf{For} $k=0,\ldots,N$ \textbf{do}
\item \quad Form $G_0=I+qH$, $G_1=I-pH$; evaluate
$\Phi(H)$ and $R(H)=\nabla\Phi(H)$.
\item \quad Diagonalize $H$ and refine any degenerate eigenspaces
to obtain a rank-one PVM $P_H$.
\item \quad Set $I_-\leftarrow I(\mathcal E,P_H)$ and
$I_+\leftarrow\Phi(H)+\Tr[R(H)^2]/(2\omega)$.
\item \quad \textbf{If} $0\leq I_+-I_-\leq\varepsilon$,
\textbf{return} $(P_H,[I_-,I_+],\mathrm{target\ met})$.
\item \quad \textbf{If} $k=N$, \textbf{return}
$(P_H,[I_-,I_+],\mathrm{target\ not\ met})$.
\item \quad Choose a Hermitian ascent direction $D$ with
$\Tr[R(H)D]>0$ (default $D=R(H)$); set
$\alpha\leftarrow\alpha_0$.
\item \quad \textbf{While} either $I+q(H+\alpha D)\not\succ0$ or
$I-p(H+\alpha D)\not\succ0$, or the feasible trial fails
$\Phi(H+\alpha D)\geq\Phi(H)+\eta\alpha\Tr[R(H)D]$,
set $\alpha\leftarrow\beta\alpha$.
\item \quad Update $H\leftarrow H+\alpha D$.
\item \textbf{End for}
\end{list}
\smallskip\hrule
\end{minipage}
\par\medskip

\prlsubheading{Implementation of the posterior-spectral receiver}

For the binary coordinates $s_0=q$ and $s_1=-p$, the curvature matrix
$C$ of Sec.~\ref{app:v3-robust} is the scalar $\omega>0$ of
Eq.~\eqref{eq:v3-certificate},
and the gradient of Eq.~\eqref{eq:v3-sm-binary-program} is
\begin{equation}
 R(H)=pq\,[\mathcal L_{I+qH}(\rho_0)-\mathcal L_{I-pH}(\rho_1)].
\end{equation}
A backtracking line search preserves
$I+qH\succ0$ and $I-pH\succ0$; the implementation uses inverse BFGS, although
the certificate is independent of the proposal method.  Matrix logarithms and
Fr\'{e}chet derivatives are evaluated in an eigenbasis using stable divided
differences. The implementation and benchmark concern full-rank inputs.
Applying the numerical method to depolarized singular inputs would require
an additional continuity bound; the exact singular theorem above is
independent of this numerical extension.

\prlsubheading{Algorithms and benchmark protocol}
The suite uses seed 20261028 and contains 24 cases at each of 17 points:
dimensions $2,4,8,16$; minimum eigenvalues $10^{-1},10^{-2},10^{-4},10^{-6}$;
priors $0.5,0.8,0.95,0.99$; and separations
$\delta=1,0.3,0.1,0.03,0.01$ in
$\rho_1=(1-\delta)\rho_0+\delta\sigma$.
The last three sweeps use $d=4$. Except in the prior sweep, priors cycle
through $0.5,0.8,0.95$.
Dimension and separation sweeps use complex Ginibre matrices $Z$ to form
$0.95ZZ^\dagger/\Tr(ZZ^\dagger)+0.05I/d$. For a prescribed minimum eigenvalue $f$, one
eigenvalue equals $f$ and the others are
$f+(1-df)x_j$, where $(x_j)$ is uniform on the $(d-2)$-simplex; the
eigenbasis is Haar random. The prior sweep uses $f=10^{-3}$.
All inputs were generated and stored before any optimization.

The compared methods are as follows.
\begin{enumerate}
 \item \emph{Posterior-spectral receiver:} feasibility-preserving inverse BFGS applied to
 Eq.~\eqref{eq:v3-sm-binary-program}, followed by spectral decomposition and
 Eq.~\eqref{eq:v3-certificate}.
 \item \emph{Direct POVM ascent:} a multistart
 mutual-information ascent over $d^2$ POVM effects, using multiplicative
 positive updates and POVM renormalization
 \cite{rehacek2005iterative}.  The starts are uniform effects, Helstrom-basis projectors and mixed-state
pretty-good measurement (PGM) effects, with the latter two split into $d^2$ effects. Each start is
mixed with $1\%$ of an independent random full-rank POVM and normalized.
This perturbation avoids the stationary identity-effect initialization.
 \item \emph{Multistart PVM:} BFGS on $U(d)$ in local exponential coordinates,
 initialized by identity, Helstrom, and Haar-random bases.
 \item \emph{Helstrom and PGM:} the minimum-error spectral measurement and the
 square-root measurement, respectively
 \cite{helstrom1976quantum,hausladen1994pretty,eldar2001square}.
 The mixed-state PGM effects are $M_i=p_i\tau^{-1/2}\rho_i\tau^{-1/2}$
 on $\supp\tau$, where the inverse is restricted to that support.
 They are positive and sum to its identity. This supplies a feasible baseline, with no guarantee of optimal information
or error.
\end{enumerate}
The posterior-spectral, direct-POVM, and multistart-PVM algorithms optimize
mutual information; the latter two supply feasible values without an
independent global upper bound. Helstrom-basis and PGM measurements are
fixed baselines, evaluated using the same information-loss bound.
The optimization budgets are fixed: at most 300 posterior iterations,
100 accepted direct-POVM iterations per start, and 30 PVM-BFGS iterations
per start. For the three optimization algorithms, the comparison therefore
measures convergence within these budgets.

\prlsubheading{Solver profiles}

Figure~\ref{fig:full-solver} reports accuracy and runtime profiles for
all 408 cases. All curves use $I_+-I(M)$ with the same float64 posterior
upper bound, rather than treating one receiver's information as the exact
optimum. At $10^{-6}$ nat, the posterior receiver succeeds in $96.8\%$
of cases, with median successful runtime $20.9$ ms; multistart PVM,
direct POVM, Helstrom-basis and PGM measurements succeed in $41.4\%$,
$20.8\%$, $0.49\%$ and $0\%$, respectively. The largest remaining
posterior bound is $6.8\times10^{-4}$ nat. Runs use one thread of an
Intel Xeon 6982P-C CPU,
Python 3.12.11, NumPy 2.5.3 and SciPy 1.18.1; timings include all starts.

\begin{figure}[!ht]
\centering
\includegraphics[width=0.95\textwidth]{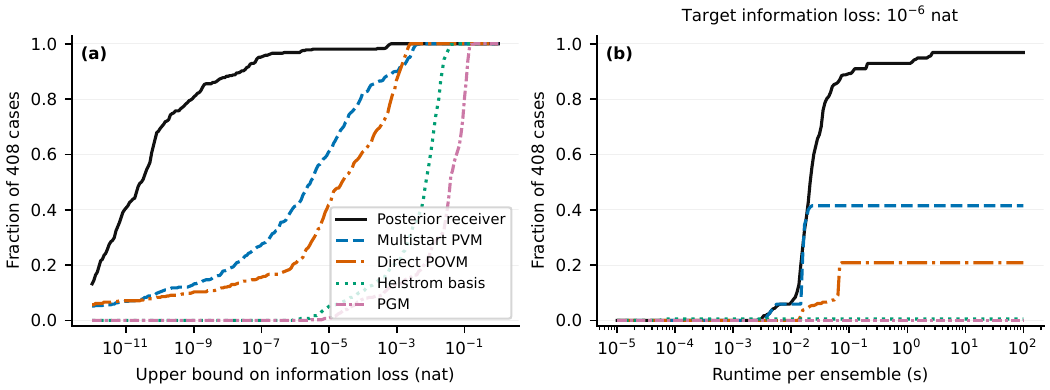}
\caption{Accuracy and runtime profiles for the 408-case suite. The
curves use a common analytic information upper bound evaluated in float64.
They give upper bounds on information loss, not interval-arithmetic certificates. The comparisons
use the fixed implementations and iteration budgets described in
Sec.~\ref{app:v3-benchmark}.}
\label{fig:full-solver}
\end{figure}

\prlappendix{Scope, prior work, and limitations}
\label{app:v3-scope}

The existence of a projective optimum for finite collinear ensembles is also a
specialization of the measured-$f$ projective-reduction theorem of Fang,
Fawzi, and Fawzi \cite{fang2026uhlmann} (Sec.~\ref{app:v3-binary}).  The proof via commuting optimized posterior operators also identifies a canonical measurement,
classifies all optimal finite-outcome measurements, bounds deviations from the
spectral blocks, and gives information bounds at finite optimization steps.
Classical minimal sufficiency and postprocessing equivalence motivate the
Blackwell interpretation \cite{blackwell1953equivalent,kuramochi2015minimal,kuramochi2018experiments}.

\prlsubheading{Orthogonal optima without commuting posterior operators}
Equation~\eqref{eq:letter-commutation} decides when $\Uaff$ is
attained, rather than whether an orthogonal measurement is optimal.
An explicit example has equal priors and Bloch vectors $(0,0,0.9)$,
$(0,0,-0.9)$, and $(0.5,0,0)$. The $\sigma_z$ measurement is optimal,
with $\Iacc=0.32975$ nat, whereas $\Uaff=0.35448$ nat,
$\chi=0.35940$ nat, and
$\max_{i,j}\norm{[G_i^\star,G_j^\star]}_F=0.75$.
Under $\sigma_z$ the third state gives equally likely outcomes.

For this example, projective optimality and strictness of the bound can
also be proved analytically. Put $a=0.9$, $s=0.5$ and
$A=[(1+a)\log(1+a)+(1-a)\log(1-a)]/3$.
A rank-one effect $w(I+\bm n\cdot\bm\sigma)$ contributes $w\phi(x,z)$,
where $x=n_x$, $z=n_z$, and
\begin{align}
 \phi(x,z)&=\frac{(1+az)\log(1+az)+(1-az)\log(1-az)}3+k(x),\\
 k(x)&=\frac{1+sx}{3}\log(1+sx)
       -\left(1+\frac{sx}{3}\right)\log\left(1+\frac{sx}{3}\right).
\end{align}
The positive even-power series of the first term bounds it by $Az^2$.
Also $k(0)=k'(0)=0$ and
$k''(x)\le s^2/[3(1-s)]$, so $k(x)\le s^2x^2/[6(1-s)]$.
Since $A\ge a^2/3>s^2/[6(1-s)]$ and $x^2+z^2\le1$,
$\phi(x,z)\le A$. Rank-one refinement and $\sum_yw_y=1$ give
$I(M)\le A$ for every POVM, with equality for the $\sigma_z$ PVM.

If $U_{\rm aff}=A$, equality in Jensen's inequality for this PVM would
require $G_0=I+a\sigma_z$, $G_1=I-a\sigma_z$, $G_2=I$.
The three states are affinely independent, so the feasible perturbation
$\delta G_2=\epsilon\sigma_x$, $\delta G_0=\delta G_1=-\epsilon\sigma_x/2$
is allowed for sufficiently small $\epsilon$.
Its derivative of $\Phi$ at $\epsilon=0$ is $s/3>0$, contradicting
optimality. Thus $U_{\rm aff}>I_{\rm acc}$ and the optimized posterior
operators do not commute, although the projective measurement is optimal.

A numerical survey illustrates how often this distinction occurs. In
200 random three-state qubit ensembles (Dirichlet priors, Bloch vectors
uniform in the ball of radius $0.95$), optimization finds a PVM matching
the best POVM value in 192. A search for violations of the dual inequality
on the Bloch sphere supports these optima; this finite search is not an
exact certificate. None has commuting $G_i^\star$ to numerical precision;
among those 192 ensembles the smallest commutator norm is $0.16$.
The remaining 8 also numerically support nonattainment: for qubits,
attainment would require an orthogonal optimum. The finite searches do not
prove its absence.

\prlsubheading{How much $\Uaff$ tightens the measured-entropy bound}
Equation~\eqref{eq:v3-sm-holevo} gives
$\Iacc\leq\Uaff\leq\sum_ip_iD_{\mathbb M}(\rho_i\Vert\tau)\leq\chi$. In the
same 200 qubit ensembles, $\Uaff$ closes a median $5\%$ (interquartile range
$1.5$--$12.6\%$) of the gap $\sum_ip_iD_{\mathbb M}(\rho_i\Vert\tau)-\Iacc$ and a
median $37\%$ of $\chi-\Iacc$. For 60 random qutrit ensembles of three states (Dirichlet priors,
Hilbert--Schmidt random states)
the medians are $8\%$ and $42\%$, and for 60 of four states $6\%$ and $25\%$.
These comparisons estimate $\Iacc$ by the best information found
with multistart POVM optimization; only the explicit three-state example
above has an analytic proof of optimality.
In these tests, $\Uaff$ improves only modestly on the average measured
relative entropy. Its main role here is to determine when the bound is
attained and to classify the optimal measurements.
Computing $\Uaff$ is a smooth concave program in $r$ Hermitian $d\times d$
matrices, and each evaluation of $\Phi$ and its gradient needs $m$
eigendecompositions of $d\times d$ matrices. We have not compared $\Uaff$
numerically with the bound of Ref.~\cite{gao2026alternative}.

\prlsubheading{Boundary and computational scope}
The stability constants for full-rank states deteriorate as an eigenvalue
approaches zero.
The closed-domain binary theorem nevertheless extends the canonical posterior
PVM and complete finite-optimizer classification to arbitrary singular states
after support compression.  A uniform ambient Frobenius-rigidity theorem does
not follow: varying-support pure states require a visibility-weighted
seminorm.  The benchmark tests convergence at fixed iteration budgets after reducing
the optimization to one Hermitian matrix. It does not determine polynomial
bit complexity or asymptotic scaling exponents.

\prlappendix{Information--error trade-off: setting and main criterion}
\label{app:extension-overview}

Consider two distinct full-rank qubit states with priors $p,1-p$, where
$0<p<1$. For mixed states, the information-optimal and minimum-error
measurements generally differ \cite{fuchs1995distinguishability,varga2024},
and information and guessing probability can rank binary measurements
differently \cite{thai2026shor}. Write $q_y=\Pr(y)$ and
$\eta_y=\Pr(0|y)$. With $h(t)=-t\log t-(1-t)\log(1-t)$, the information
and the error of guessing the more likely state are
\begin{align}
 I(M)&=h(p)-\sum_yq_yh(\eta_y),\nonumber\\
 \Pe(M)&=\sum_yq_y\min(\eta_y,1-\eta_y).
 \label{eq:arxiv-posterior-scores}
\end{align}
This posterior representation is also used in Ref.~\cite{varga2024}.
Both quantities concern inference about the input, rather than disturbance
of the measured state. Every outcome leads to a guess; none is discarded or
declared inconclusive.

We maximize $F_\lambda=I-\lambda\Pe$ for $\lambda\geq0$, using natural
logarithms. At $\lambda=0$ this is accessible information; as
$\lambda\to\infty$, the limiting task minimizes error, with information
breaking ties. Let $\VPOVM(\lambda)$ and $\VPVM(\lambda)$
be the maxima over generalized and orthogonal measurements, respectively,
and define
\begin{equation}
 \Delta_\lambda=\VPOVM(\lambda)-\VPVM(\lambda).
 \label{eq:arxiv-three-outcome-gap}
\end{equation}
\prlsubheading{Weighted and fixed-error formulations} Let $I_{\rm POVM}(\varepsilon)$ and
$I_{\rm PVM}(\varepsilon)$ be the largest information at error probability
at most $\varepsilon$ over generalized measurements and over orthogonal
measurements and their random mixtures, with the basis label kept. Keeping
the label makes information and error linear under mixing, so both
achievable sets are convex and both frontiers are concave and
nondecreasing; $V_{\rm POVM}$ and $V_{\rm PVM}$ are their Legendre
transforms, $V(\lambda)=\max_\varepsilon[I(\varepsilon)-\lambda\varepsilon]$.
Put $\Delta I(\varepsilon)=I_{\rm POVM}(\varepsilon)-I_{\rm PVM}(\varepsilon)$.
If the POVM optimum at weight $\lambda$ has error $\varepsilon_1$, then
$V_{\rm POVM}(\lambda)\leq I_{\rm POVM}(\varepsilon_1)-\lambda\varepsilon_1$
and $V_{\rm PVM}(\lambda)\geq I_{\rm PVM}(\varepsilon_1)-\lambda\varepsilon_1$,
so $\Delta_\lambda\leq\Delta I(\varepsilon_1)$. Conversely, $\Delta I$
vanishes at the Helstrom and maximum-information errors, where both
endpoint optima are projective; at an interior maximizer $\varepsilon_0$, a
supergradient $\lambda_0\geq0$ of $I_{\rm PVM}$ gives
$V_{\rm PVM}(\lambda_0)=I_{\rm PVM}(\varepsilon_0)-\lambda_0\varepsilon_0$
and $V_{\rm POVM}(\lambda_0)\geq I_{\rm POVM}(\varepsilon_0)-\lambda_0\varepsilon_0$.
Hence $\max_\varepsilon\Delta I(\varepsilon)=\max_\lambda\Delta_\lambda=\Delta^*$:
the largest weighted gap is the largest information gain at a fixed error
probability over all orthogonal measurements and their mixtures.
Numerically the two sides agree to within $10^{-7}$ nat for the ensembles
of Sec.~\ref{app:pareto-size}.
Against a single, unrandomized PVM the gain can be larger, because the
single-PVM frontier need not be concave; randomized PVMs are the fair
baseline.

The extreme posteriors are the eigenvalues of
$Q=\tau^{-1/2}p\rho_0\tau^{-1/2}$, where
$\tau=p\rho_0+(1-p)\rho_1$.

\hypertarget{thm:one}{}\textit{Forced-switch criterion}---Suppose the states
do not commute and the extreme posteriors straddle $1/2$, so that the
Helstrom error $\Pe^{\mathrm H}$ is smaller than the prior error
$\Pe^0=\min(p,1-p)$. If both outcomes of the information-optimal PVM
strictly favor the same state, then no choice of optimal PVMs varies
continuously with $\lambda$: the best basis switches abruptly at some
$\lambda_J>0$. At every switch a three-outcome POVM strictly outperforms
every PVM and every random mixture of PVMs. The weighted gain persists over an open interval of penalties by continuity.
At the error of the optimal POVM at the switch, the fixed-budget gain is also
strict. This budget is interior to the feasible error range up to the
maximum-information endpoint, since the two frontiers agree at both
endpoints. Continuity of the frontiers there gives an open interval of
error budgets with a strict gain. The strict preference
has the angle form in Eq.~\eqref{eq:pareto-condC}; it is sufficient but not
necessary. Equal priors never satisfy it, yet can still show gaps.

The scalar formulation and its saturation are proved in
Sec.~\ref{app:pareto-scalar}. Section~\ref{em:three} proves that three
outcomes suffice and gives global contact certificates. The centered
common optimum and the exact existence criterion at fixed posterior
spectrum are established in Secs.~\ref{app:pareto-centered}
and~\ref{app:pareto-construct}. Theorems~\ref{thm:pareto-tie}
and~\ref{thm:pareto-forced}, with Corollaries~\ref{cor:pareto-tie}
and~\ref{cor:pareto-onset}, prove the switch criterion and locate the onset
of the gap. Numerical certificates are collected in
Sec.~\ref{app:pareto-size};
Sec.~\ref{app:pareto-physical} gives the operational interpretation and
physical examples. Unequal decision costs and other information scores are
treated in Sec.~\ref{app:pareto-general}.

\prlappendix{Scalar formulation and saturation}
\label{app:pareto-proofs}
\label{app:pareto-scalar}
We derive an exact scalar optimization for the task of
Sec.~\ref{app:extension-overview}, including the reduction to one real
semicircle. The two qubit states remain distinct and positive definite,
$0<p<1$, and $\lambda\geq0$. The proofs below do not assume that penalties
with a nonprojective advantage form one connected interval.
Set $A=p\rho_0$, $B=(1-p)\rho_1$, $\tau=A+B$ and
$Q=\tau^{-1/2}A\tau^{-1/2}$. In the posterior normal form,
\begin{equation}
 Q=cI+d\sigma_z,\quad d>0,\quad 0<c-d<c+d<1,\qquad
 \tau=\tfrac12(I+x_\tau\sigma_x+z_\tau\sigma_z),\quad
 x_\tau\geq0,\quad x_\tau^2+z_\tau^2<1.
\end{equation}
Conversely, every such $Q,\tau$ defines a full-rank ensemble by
\begin{equation}
 p=\operatorname{tr}(\tau Q),\qquad
 \rho_0=\frac{\tau^{1/2}Q\tau^{1/2}}p,\qquad
 \rho_1=\frac{\tau^{1/2}(I-Q)\tau^{1/2}}{1-p}.
 \label{eq:inverse-ensemble}
\end{equation}
The normalized posterior operator $Q$ is distinct from the optimized
matrix $H_\star$ used above. The representation below describes all
rank-one measurements; it does not imply that the eigenbasis of $Q$
maximizes information.

The following normalized-effect transformation is given in
Ref.~\cite[Eqs.~(13)--(19)]{varga2024}. We use it to formulate the weighted
information--error objective and establish its outcome-count and spectral
existence criteria.
For each rank-one effect write
\begin{equation}
 \tau^{1/2}M_y\tau^{1/2}=q_y|\psi_y\rangle\langle\psi_y|.
\end{equation}
Then $q_y$ is the outcome probability, its posterior is
$\eta_y=\langle\psi_y|Q|\psi_y\rangle$, and completeness is equivalent
to $\sum_yq_y|\psi_y\rangle\langle\psi_y|=\tau$.
The inverse is $M_y=q_y\tau^{-1/2}|\psi_y\rangle\langle\psi_y|\tau^{-1/2}$.
Here $q_y=\Tr(\tau M_y)>0$ for each nonzero effect. Conversely, the
displayed inverse maps every pure-state decomposition of $\tau$ to positive
rank-one effects summing to $I$; taking traces gives $\sum_yq_y=1$.
The posterior formula follows from $\Tr(AM_y)=q_y\langle\psi_y|Q|\psi_y\rangle$.

Define $s(z)=\sqrt{1-z^2}$ and
$g_\lambda(z)=-h(c+dz)+\lambda|c+dz-1/2|$.
The maximum of the weighted objective equals $h(p)-\lambda/2$ plus
\begin{equation}
 \max_{\mu\in\mathcal P([-1,1])}
 \left\{\int g_\lambda\,\dd\mu:
 \int z\,\dd\mu=z_\tau,\quad \int s\,\dd\mu\geq x_\tau\right\}.
 \label{eq:scalar}
\end{equation}
Necessity follows from $n_x\leq s(n_z)$ for every pure Bloch vector.
For sufficiency let $S=\int s\,\dd\mu$. If $S>0$, split the atom at each
$z$ between $(s(z),0,z)$ and $(-s(z),0,z)$ in the common relative proportions
$(1+x_\tau/S)/2$ and $(1-x_\tau/S)/2$. The resulting average is
$(x_\tau,0,z_\tau)$ and the score is unchanged. If $S=0$, then $x_\tau=0$ and the
two poles suffice. The inverse correspondence gives a POVM.
Spectral refinement of a general effect cannot reduce the score, since
$-h(\eta)-\lambda\min(\eta,1-\eta)$ is convex and the unrefined posterior
is the probability-weighted mean of its refinements. Thus no general POVM
is excluded by the scalar optimization.

The inequality $\int s\,\dd\mu\geq x_\tau$ in Eq.~\eqref{eq:scalar}, which
requires the outcomes to reach horizontally as far as $\tau$, saturates at
every optimum.
Let $\nu$ be the two-pole distribution with mean $z_\tau$. Strict convexity of
$g_\lambda$ gives $\int g_\lambda\,\dd\nu>\int g_\lambda\,\dd\mu$
whenever $\mu$ has interior mass. If $S>x_\tau$, replace $\mu$ by
\begin{equation}
 \widetilde\mu=\frac{x_\tau}{S}\mu+
       \left(1-\frac{x_\tau}{S}\right)\nu.
\end{equation}
Its mean is unchanged, $\int s\,\dd\widetilde\mu=x_\tau$, and its score strictly
increases. At saturation, equality in $\mathbb E n_x\leq\mathbb E s$
requires $n_x=s(z)$ at every occupied atom. All optimal rank-one atoms
therefore lie on the positive real semicircle, including its poles.

\prlappendix{Outcome count and global contact certificates}
\label{em:three}
As in the accessible-information problem for real qubit ensembles
\cite{davies1978information,sasakietal1999accessible}, three outcomes
suffice. Here the saturation proof above supplies the required planar
reduction for the weighted objective.
The saturated scalar problem has three linear moments, $1,z,s(z)$.
An extreme feasible measure has at most three atoms: otherwise four
disjoint positive-mass support sets yield a linear dependence among their
three moment vectors, allowing a nonzero signed perturbation in both
directions while preserving positivity and all moments. Compactness supplies
an extreme maximizer of the continuous linear objective.
Two rank-one qubit effects summing to $I$ are orthogonal projections.
Their pure-state decomposition images form a chord through $(x_\tau,z_\tau)$ and need not be antipodal.
Moreover, every two-effect POVM $\{E,I-E\}$ is a stochastic postprocessing
of the eigenbasis measurement of $E$. Such postprocessing cannot increase
information or decrease Bayes error. Hence the best two-outcome value
equals the best PVM value.

For $\Delta_\lambda=V_{\rm POVM}(\lambda)-V_{\rm PVM}(\lambda)$ this gives
\begin{equation}
 \Delta_\lambda>0\quad\Longleftrightarrow\quad
 \text{the optimal minimal number of outcomes is three}.
\end{equation}
An optimizer can have more labeled effects through redundant splitting,
but three outcomes suffice to maximize each weighted objective.
This statement does not impose the same outcome bound on every performance
pair along a flat segment of the frontier.

A useful dual of Eq.~\eqref{eq:scalar} is
\begin{equation}
 \min_{\alpha,\beta\in\mathbb R,\,\kappa\geq0}
 \left\{\alpha+\beta z_\tau-\kappa x_\tau:
 \alpha+\beta z-\kappa s(z)\geq g_\lambda(z)\ \forall z\right\}.
 \label{eq:dual}
\end{equation}
The point mass at $z_\tau$ satisfies the inequality constraint strictly,
since $s(z_\tau)>x_\tau$. Equivalently a supporting-plane construction for the compact
concave roof gives strong duality. A feasible upper bound touching the objective at exactly three
noncollinear points proves a strict advantage whenever the average state
lies inside their triangle. The bound is attained by assigning positive
weight to all three points. A distribution on any two of them cannot
reproduce an interior mean. Compactness of the PVM family then makes
its shortfall strictly positive.
For noncommuting signals and $\lambda>0$, Lemmas~\ref{lem:pareto-contacts}
and~\ref{lem:pareto-unique} below also show that optimal contacts avoid
the decision height $k=(1/2-c)/d$. Thus every outcome favors a definite
guess. Two contacts on the same side of $k$ distinguish confidence levels
within that guess. A rank-one qubit POVM with three outcomes can be
implemented by a projective measurement with one additional level
\cite{peres1993quantum}.

\prlappendix{Centered spectra and a common endpoint optimum}
\label{app:pareto-centered}
Suppose $c=1/2$. Set $a=2d\in(0,1)$ and
\begin{equation}
 H_a(u)=h\left(\frac{1+a\sqrt{1-u^2}}2\right),\qquad
 \phi_\lambda(u)=-H_a(u)+\lambda d\sqrt{1-u^2}.
\end{equation}
Then $g_\lambda(z)=\phi_\lambda(s(z))$. For $0<u<1$, with
$t=\sqrt{1-u^2}$,
\begin{equation}
 H_a'(u)=\frac{au}{t}\operatorname{atanh}(at)>0,\qquad
 H_a''(u)=\frac a{t^3}\left[
 \operatorname{atanh}(at)-\frac{at(1-t^2)}{1-a^2t^2}\right]>0.
\end{equation}
The last sign follows from
$\operatorname{atanh}(at)>at\geq at(1-t^2)/(1-a^2t^2)$.
Thus $\phi_\lambda$ is strictly concave and decreasing, so
\begin{equation}
 \mathbb E g_\lambda(z)
 \leq\phi_\lambda(\mathbb E s(z))\leq\phi_\lambda(x_\tau).
\end{equation}
Let $z_0=\sqrt{1-x_\tau^2}$. The two states
$(x_\tau,0,\pm z_0)$ with probabilities $(1\pm z_\tau/z_0)/2$ attain equality
for every $\lambda$. Both probabilities are positive, and the inverse
correspondence gives a PVM. Its posterior values are $1/2\pm dz_0$; consequently
\begin{equation}
 \Pe^{\min}=\tfrac12-dz_0,\qquad
 I_{\max}=h(\tfrac12+dz_\tau)-h(\tfrac12+dz_0).
\end{equation}
This one performance point dominates all others.

For completeness, the invariant centered condition is
\begin{equation}
 \operatorname{tr}[(A+B)^{-1}A]=1
 \quad\Longleftrightarrow\quad \det A=\det B
 \quad\Longleftrightarrow\quad p^2\det\rho_0=(1-p)^2\det\rho_1.
\end{equation}
To verify it, write $\det(tA+B)=t^2\det A+tC+\det B$ and differentiate
at $t=1$. The trace is $(2\det A+C)/(\det A+C+\det B)$.
This characterizes centered spectra, not all possible common-optimum
ensembles; commuting ensembles provide other common-optimum examples.
In Bloch-vector form the centered condition reads
$p^2(1-r_0^2)=(1-p)^2(1-r_1^2)$, where $r_i$ is the Bloch length of
$\rho_i$. It is the fourth explicitly solvable case, Eq.~(3.255), of
Fuchs's variational analysis of the information-optimal qubit basis, where
that basis was observed to coincide with the Helstrom basis
\cite{fuchs1995distinguishability}; for equal priors it reduces to equal
Bloch radii, the coincidence condition reported in
Ref.~\cite{varga2024}. A common endpoint optimum also optimizes every
nonnegative weighted sum of the two endpoint scores; the argument above
constructs this measurement explicitly.

\prlappendix{Constructive spectral existence theorem}
\label{app:pareto-construct}
For a fixed posterior spectrum $0<c-d<c+d<1$, there exist a full-rank
average state $\tau$ and a weight $\lambda>0$ with
$V_{\rm POVM}(\lambda)>V_{\rm PVM}(\lambda)$ if and only if
$c-d<1/2<c+d$ and $c\ne1/2$. When these conditions hold, the advantage
occurs for an open set of average states and persists over an open interval
of weights for each such state. This is an existence criterion at fixed
$Q$, not a claim that every compatible $\tau$ exhibits an advantage.

Fix $Q$ with $0<c-d<1/2<c+d<1$ and $c\ne1/2$. We construct a slope
$\lambda_0>0$ and an open set of average states with a strict three-outcome
advantage. Reflect labels and the posterior axis if necessary so $c<1/2$,
and put $k=(1/2-c)/d\in(0,1)$. Define
\begin{equation}
 R(z)=\frac{d^2s(z)^3}{(c+dz)(1-c-dz)}.
\end{equation}
It has a unique strict maximum $K_*$ at some $m<k$.
Indeed, with $\eta=c+dz$, $a=c-d$ and $b=c+d$,
\begin{equation}
 \frac{\dd^2\log R}{\dd\eta^2}
 =-\frac32\left[(\eta-a)^{-2}+(b-\eta)^{-2}\right]
   +\eta^{-2}+(1-\eta)^{-2}<0.
\end{equation}
The sign follows termwise from $a>0$, $b<1$. Also $R$ vanishes at both
ends and $(\log R)'(k)=-3k/(1-k^2)<0$, where this last derivative is in $z$.

For $0<\kappa<K_*$ write
\begin{equation}
 F_\kappa(z)=h(c+dz)-\kappa s(z),\qquad
 F_\kappa''(z)=\frac{\kappa-R(z)}{s(z)^3}.
\end{equation}
If $u<v$ are the two roots of $R=\kappa$, then $F_\kappa$ is convex on
the two outer intervals and concave between them, with derivatives tending
to $-\infty$ and $+\infty$ at the respective endpoints.
For each $t\in(F_\kappa'(v),F_\kappa'(u))$ there are unique outer
stationary points $z_1(t)<u$ and $z_2(t)>v$ of $F_\kappa-tz$.
Their depth difference
\begin{equation}
 L(t)=[F_\kappa(z_2)-tz_2]-[F_\kappa(z_1)-tz_1],\qquad
 L'(t)=z_1-z_2<0
\end{equation}
is positive at the lower slope endpoint and negative at the upper one.
These signs follow by integrating $F_\kappa'-t$ between the limiting
outer stationary points. Hence there is one slope for which the two minima have equal value.
The corresponding common tangent $\ell(z)=\ell_0+t_\kappa z$ lies
below $F_\kappa$ on the whole interval and touches it only at $z_1,z_2$.

As $\kappa\uparrow K_*$ both contacts approach $m$.
At $K_*$ the function is strictly convex, its second derivative vanishing
only at $m$; a limiting support cannot have two separated contacts.
Since $z_1<u<v<z_2$ and $u,v\to m$, the common limit is $m$.
Choose $\kappa$ close enough that $z_2<k$ and put $D_0=F_\kappa-\ell$.
Then $D_0\geq0$, with exactly those two zeros, and $D_0''>0$ on $[k,1)$.
The error penalty is affine on each side of $k$. Its contribution below
$k$ can therefore be absorbed into the supporting line, leaving only
$2\lambda d(z-k)$ above $k$. Raising $\lambda$ preserves the two lower
contacts until a third contact appears above the decision height. The
following construction determines that first contact exactly.
The function
\begin{equation}
 T(C)=D_0(C)-(C-k)D_0'(C)
\end{equation}
has $T(k)>0$, $T'(C)=-(C-k)D_0''(C)<0$, and $T(C)\to-\infty$
as $C\to1$. The last limit uses finite $D_0(1)$ and
$D_0'(1-)=+\infty$. There is a unique root $z_3\in(k,1)$, and
$D_0'(z_3)>0$ because $D_0(z_3)>0$.
Set $\lambda_0=D_0'(z_3)/(2d)>0$ and
\begin{equation}
 D(z)=\begin{cases}
 D_0(z),&z\leq k,\\
 D_0(z)-2\lambda_0d(z-k),&z\geq k.
 \end{cases}
\end{equation}
The right piece is strictly convex and has a double zero at $z_3$;
the left piece has exactly the zeros $z_1,z_2$. Thus $D\geq0$ globally,
with exactly three zeros. For
$\alpha=\lambda_0dk-\ell_0$ and $\beta=-\lambda_0d-t_\kappa$,
\begin{equation}
 D(z)=\alpha+\beta z-\kappa s(z)-g_{\lambda_0}(z).
\end{equation}
This is a global dual certificate. Since $n_x\leq s(z)$ and $\kappa>0$,
it also bounds all complex pure-state atoms, with equality only at
$(s(z_i),0,z_i)$ for $i=1,2,3$.

Choose any average-state point in the open interior of their triangle.
Its barycentric weights are strictly positive, and Eq.~\eqref{eq:inverse-ensemble}
gives full-rank signal states. The contact criterion proves a strictly
positive gap at $\lambda_0$. The optimal POVM and PVM values are both Lipschitz in $\lambda$, because
$0\leq\Pe\leq\min(p,1-p)$. The gap therefore persists on an open interval.
The triangle gives an open set in the normal-form half disk; rotating about
the posterior axis gives an open set of average states in the full Bloch ball.

For the converse, centered spectra have the common optimum just proved.
If the spectrum lies entirely on one side of $1/2$, every posterior does
also, so every measurement has error $\min(p,1-p)$. The established
binary-qubit information theorem \cite{keil2024} supplies a PVM that
maximizes the information and hence, at constant error, every weighted
score. This completes the spectral ``if and only if'' assertion,
including spectra that merely touch $1/2$.

\prlappendix{Basis switches and the three-outcome advantage}
\label{app:pareto-switch}
Throughout this section the two states are full-rank and do not commute,
and $\lambda\geq0$. For commuting states ($x_\tau=0$) optimal duals can
have $\kappa=0$ and contacts at the poles, and the results below are not
claimed for them. Write $\rho_i=(I+\bm r_i\cdot\bm\sigma)/2$ and
$\bm\gamma=p\bm r_0-(1-p)\bm r_1$, so that
$A-B=\frac12[(2p-1)I+\bm\gamma\cdot\bm\sigma]$. Let
$\Pe^{0}=\min(p,1-p)$ be the error of guessing without measuring and
$\Pe^{\mathrm H}$ the Helstrom error. As before, $\eta=c+dz$ and $k=(1/2-c)/d$.
For $(\alpha,\beta,\kappa)$ put
\begin{equation}
 D(z)=\alpha+\beta z-\kappa s(z)-g_\lambda(z),
 \label{eq:pareto-D}
\end{equation}
so that $(\alpha,\beta,\kappa)$ is feasible for Eq.~\eqref{eq:dual} exactly
when $\kappa\geq0$ and $D\geq0$ on $[-1,1]$. Noncommuting is equivalent to
$x_\tau>0$. Indeed, $x_\tau=0$ means $[\tau,Q]=0$, since $Q$ has distinct
eigenvalues. Then $A=\tau^{1/2}Q\tau^{1/2}$ commutes with $\tau$, hence
with $B=\tau-A$. Conversely, $[A,B]=0$ gives $[A,\tau]=0$, so
$Q=\tau^{-1/2}A\tau^{-1/2}$ commutes with $\tau$. The log-concavity proof
for $R$ in Sec.~\ref{app:pareto-construct} uses only $0<c-d<c+d<1$, so $R$
is strictly log-concave on $(-1,1)$ for every spectrum. Off $z=k$,
\begin{equation}
 D''(z)=\frac{\kappa-R(z)}{s(z)^3},
 \label{eq:pareto-Dpp}
\end{equation}
and the set $\{R>\kappa\}$ is an open interval $(u,v)$, empty if
$\kappa\geq\max R$.

\begin{lemma}[Dual attainment]\label{lem:pareto-dual}
Let $x_\tau>0$. The dual Eq.~\eqref{eq:dual} has an optimal solution, and its
value equals the maximum in Eq.~\eqref{eq:scalar}. For every optimal $\mu$
and every optimal dual, $\mu$ is supported on $\{D=0\}$ and
$\kappa(\int s\,\dd\mu-x_\tau)=0$.
\end{lemma}

\par\textit{Proof}---
For $(a,b)$ in the moment set
$\mathcal M=\{(a,b):|a|\leq1,\,b\leq s(a)\}$ let $v(a,b)$ be the maximum of
$\int g_\lambda\,\dd\mu$ subject to $\int z\,\dd\mu=a$ and
$\int s\,\dd\mu\geq b$. The set $\mathcal M$ is the feasible region, because
$s$ is concave and the point mass at $a$ has $\int s=s(a)$. Mixing
feasible measures shows that $v$ is concave, and $v$ is nonincreasing in
$b$ and bounded above by $\max g_\lambda$. The point $(z_\tau,x_\tau)$ is
interior to $\mathcal M$, since $|z_\tau|<1$ and $x_\tau<s(z_\tau)$. A
concave function finite near an interior point has a supergradient there:
there are $\beta$ and $\kappa$ with
$v(a,b)\leq v(z_\tau,x_\tau)+\beta(a-z_\tau)-\kappa(b-x_\tau)$ on $\mathcal M$.
Taking $a=z_\tau$ and $b<x_\tau$ gives $\kappa\geq0$. Taking the point mass
at $z$, with $(a,b)=(z,s(z))$, gives $g_\lambda(z)\leq\alpha+\beta z-\kappa
s(z)$, where $\alpha=v(z_\tau,x_\tau)-\beta z_\tau+\kappa x_\tau$. This
dual point is therefore feasible, and its value $v(z_\tau,x_\tau)$ is
optimal by weak duality,
$\int g_\lambda\,\dd\mu\leq\alpha+\beta z_\tau-\kappa\int s\,\dd\mu\leq
\alpha+\beta z_\tau-\kappa x_\tau$. For optimal $\mu$ and an optimal dual both
inequalities are equalities. Thus $\int D\,\dd\mu=0$ with $D\geq0$
continuous, and $\kappa(\int s\,\dd\mu-x_\tau)=0$.

\begin{lemma}[Contact count]\label{lem:pareto-contacts}
Let $D\geq0$ with $\kappa>0$. Then $D$ does not vanish at $z=\pm1$, and if
$\lambda>0$ it does not vanish at $z=k$. Every zero $z_i$ satisfies
$D'(z_i)=0$ and $R(z_i)\leq\kappa$. If $\lambda>0$ and $-1<k<1$, $D$ has at most
three zeros; otherwise it has at most two. If it has three, exactly two lie
on the same side of $k$, and $R>\kappa$ somewhere between those two.
\end{lemma}

\par\textit{Proof}---
On $[-1,1]\setminus\{k\}$ the function $g_\lambda$ is analytic, since
$0<\eta<1$. Near $z=1$,
$D'(z)=\beta+\kappa z/s(z)-g_\lambda'(z)\to+\infty$, so $D$ increases into
the pole and $D(1)=0$ would force $D<0$ just below $1$. The pole $z=-1$ is
symmetric. At $k$,
$D'(k^+)-D'(k^-)=-2\lambda d<0$, while a zero of a nonnegative function needs
$D'(k^-)\leq0\leq D'(k^+)$. Every remaining zero is an interior minimum of a
$C^2$ function, so $D'=0$ and $D''\geq0$ there, and
Eq.~\eqref{eq:pareto-Dpp} gives $R\leq\kappa$. Zeros therefore lie in
$[-1,u]\cup[v,1]$, or anywhere in $[-1,1]$ if $(u,v)$ is empty. Cut these
intervals at $k$. On each resulting piece $D''>0$ except at isolated points,
so $D$ is strictly convex there, and a strictly convex nonnegative function
vanishes at most once on an interval. There are at most two pieces if
$\lambda=0$, if $k\notin(-1,1)$, or if $k\in[u,v]$, and at most three
otherwise. With three zeros, $k$ cuts one of $[-1,u]$, $[v,1]$; the zero in
the uncut interval and the zero on the near side of $k$ are separated by
$(u,v)$, and the third lies across $k$.

\begin{lemma}[Unique optimal measure]\label{lem:pareto-unique}
Let $x_\tau>0$. Every optimal dual has $\kappa>0$, and the maximizer $\mu^*$
of Eq.~\eqref{eq:scalar} is unique. It satisfies $\int s\,\dd\mu^*=x_\tau$,
and its atoms are zeros of $D$: at most three, none at a pole, and, for
$\lambda>0$, none at $k$. At $\lambda=0$ it has exactly two atoms, so the
information-optimal PVM is unique. Every optimal rank-one POVM has its
normalized effects at these atoms on the positive real semicircle, with
these total weights.
\end{lemma}

\par\textit{Proof}---
If $\kappa=0$, then $D''=-g_\lambda''<0$ off $k$ and the kink at $k$ is
downward, so $D$ is strictly concave. A nonnegative strictly concave function
vanishes at most at $\pm1$. An optimal measure exists by compactness, and
by Lemma~\ref{lem:pareto-dual} it would be supported on $\{\pm1\}$, giving
$\int s\,\dd\mu=0<x_\tau$. Hence $\kappa>0$. Lemma~\ref{lem:pareto-dual} then gives
saturation, and Lemma~\ref{lem:pareto-contacts} gives at most three
support points with the stated locations. The weights solve
$\sum_iw_i(1,z_i,s(z_i))=(1,z_\tau,x_\tau)$. Three distinct points of a circle
are not collinear, so the vectors $(1,z_i,s(z_i))$ are linearly independent
and the weights are unique. At $\lambda=0$ there is no kink, so there are at
most two atoms. One atom would put $\tau$ on the sphere, so there are exactly
two, and they form a chord: a PVM. The last sentence follows from the
saturation argument after Eq.~\eqref{eq:scalar}.

By the lemma, optimal POVMs are unique only up to relabeling and splitting
outcomes (replacing $M_y$ by $tM_y$ and $(1-t)M_y$, $0<t<1$). Indeed,
refining an effect into rank-one pieces cannot lower $F_\lambda$, and it
raises $F_\lambda$ strictly unless all pieces share one posterior, because
$-h(\eta)-\lambda\min(\eta,1-\eta)$ is strictly convex; pieces with the same posterior have the same height. The refined POVM is also optimal, so saturation of the $s$-moment places every piece on the positive real semicircle, where $z$ fixes the state. The pieces are therefore proportional, and every optimal effect is a multiple of one atom's rank-one operator.

\begin{theorem}[A tie forces a gap]\label{thm:pareto-tie}
Let $x_\tau>0$ and $\lambda_J>0$. If two distinct PVMs, that is, two
different unordered bases, both attain $\VPVM(\lambda_J)$, then
$\Delta_{\lambda_J}>0$, and the optimal measure $\mu^*$ has exactly three
atoms.
\end{theorem}

\par\textit{Proof}---
A PVM maps to two normalized points $\psi_{1,2}$ on the Bloch sphere with
$q_1\psi_1+q_2\psi_2=\tau$, and hence to the measure
$q_1\delta_{z_1}+q_2\delta_{z_2}$. This measure is feasible for
Eq.~\eqref{eq:scalar} with score $F_\lambda-h(p)+\lambda/2$, since
$\int s\,\dd\mu\geq q_1n_{x,1}+q_2n_{x,2}=x_\tau$. If $\Delta_{\lambda_J}=0$,
both measures are optimal. By Lemma~\ref{lem:pareto-unique} they saturate,
so all four points lie on the positive real semicircle, where $z$
determines the point. Distinct PVMs have distinct effects, and the inverse
map $M_y=q_y\tau^{-1/2}|\psi_y\rangle\langle\psi_y|\tau^{-1/2}$ shows that
they have distinct point pairs. They therefore give two distinct optimal
measures, contradicting uniqueness. Hence $\Delta_{\lambda_J}>0$, so
$\mu^*$ is not a chord; since a saturated two-atom measure is a chord and a
single atom would put $\tau$ on the sphere, $\mu^*$ has exactly three atoms.

\begin{corollary}\label{cor:pareto-tie}
Let $x_\tau>0$.
(a) Every tie weight $\lambda_J$ lies in the open set
$\{\lambda>0:\Delta_\lambda>0\}$.
(b) Wherever $\Delta_\lambda=0$, the optimal PVM is unique and varies
continuously. A choice of optimal PVMs can be discontinuous only at weights
with a gap.
(c) $\VPOVM$ is $C^1$ (with a right derivative at $\lambda=0$), with
$\VPOVM'(\lambda)=-\Pe(\mu^*_\lambda)$. The optimal measure is weakly
continuous in $\lambda$, and its error and information are continuous.
The number of outcomes can change only by a weight vanishing or two atoms merging.
The one-sided derivatives of $\VPVM$ at $\lambda_J$
are $-\max\Pe$ (left) and $-\min\Pe$ (right) over the tied PVMs. Hence
$\Delta'(\lambda_J^+)-\Delta'(\lambda_J^-)=-(\max\Pe-\min\Pe)$: $\Delta$ has a
concave corner at $\lambda_J$, strict when the tied PVMs differ in error.
(d) Among the PVMs tied at $\lambda_J$, those of minimum and maximum error
span a segment of an exposed face of the convex hull of single-PVM
performance points; more than two PVMs may tie. If the two differ in
$(\Pe,I)$, points inside that segment that no tied PVM attains are reached
only by random mixtures, so the single-PVM frontier is not concave there.
A POVM lies strictly above the line of slope $\lambda_J$ through this edge,
and above every randomized-PVM performance point at the same error.
\end{corollary}

\par\textit{Proof}---
(a) Both value functions are Lipschitz in $\lambda$.
(b) The argmax is upper hemicontinuous by Berge's maximum theorem
\cite[Ch.~16]{aliprantis1999infinite}; uniqueness gives continuity.
(c) This follows from Danskin's theorem \cite{danskin1966theory}, because the convex function
$\VPOVM-h(p)+\lambda/2=\max_\mu[\int(-h)\,\dd\mu+\lambda\int|\eta-\tfrac12|\,\dd\mu]$
has a unique maximizer. Because $\mu^*_\lambda$ is unique for every
$\lambda$, Berge's theorem makes it weakly continuous in $\lambda$, so the
error and information of the optimum are continuous and $\VPOVM$ is $C^1$
(with a right derivative at $\lambda=0$). Directly, for $t>0$, optimality at
$\lambda$ and $\lambda+t$ gives
$-\Pe(\mu^*_\lambda)\leq[\VPOVM(\lambda+t)-\VPOVM(\lambda)]/t
\leq-\Pe(\mu^*_{\lambda+t})$.
Compactness and uniqueness force every limit of maximizers to be the
maximizer at the limiting weight, so the two bounds converge to the same
value. The left derivative follows in the same way for $\lambda>0$.
(d) By Theorem~\ref{thm:pareto-tie} some POVM lies strictly
above the line of slope $\lambda_J$ through the edge, that is, above every
random mixture of PVMs.

\begin{lemma}[Error of a PVM]\label{lem:pareto-pvm-error}
The PVM with Bloch axis $\bm n$ has
\begin{equation}
 \Pe=\tfrac12-\tfrac12\max\bigl(|2p-1|,|\bm\gamma\cdot\bm n|\bigr),
 \label{eq:pareto-pvm-error}
\end{equation}
and its two posteriors obey
$\operatorname{sgn}(\eta_\pm-\tfrac12)=\operatorname{sgn}(2p-1\pm\bm\gamma\cdot\bm n)$.
The spectrum straddles $1/2$ if and only if $|\bm\gamma|>|2p-1|$. Then
$\Pe^{\mathrm H}=\frac12(1-|\bm\gamma|)<\Pe^{0}$, and for an axis at angle
$\omega\in[0,\pi/2]$ to the line of $\bm\gamma$,
$\Pe=\tfrac12-\tfrac12\max[1-2\Pe^{0},(1-2\Pe^{\mathrm H})\cos\omega]$.
\end{lemma}

\par\textit{Proof}---
For $M_\pm=\frac12(I\pm\bm n\cdot\bm\sigma)$,
$\Tr[(A-B)M_\pm]=\frac12(2p-1\pm\bm\gamma\cdot\bm n)$, and
$\eta_\pm-\tfrac12=\Tr[(A-B)M_\pm]/(2q_\pm)$. Moreover,
$\Pe=\sum_\pm\min(\Tr AM_\pm,\Tr BM_\pm)=\frac12-\frac12\sum_\pm|\Tr(A-B)M_\pm|$,
and $|x+y|+|x-y|=2\max(|x|,|y|)$. The spectrum straddles exactly when
$Q-I/2=\frac12\tau^{-1/2}(A-B)\tau^{-1/2}$ is indefinite. By Sylvester's law
of inertia this holds exactly when $A-B$, with eigenvalues
$\frac12(2p-1\pm|\bm\gamma|)$, is indefinite. Minimizing
Eq.~\eqref{eq:pareto-pvm-error} over $\bm n$ gives the Helstrom error, and
$|2p-1|=1-2\Pe^{0}$.

\textit{The condition.} Assume straddling, so that
$\cos\omega_c=(1-2\Pe^{0})/(1-2\Pe^{\mathrm H})\in[0,1)$ defines
$\omega_c\in(0,\pi/2]$. Let $\delta\theta$ be the angle between the axis of
the information-optimal PVM $P_0$ and the line of $\bm\gamma$. By
Lemma~\ref{lem:pareto-pvm-error} the following are equivalent:
\begin{equation}
 \delta\theta>\omega_c
 \iff|\bm\gamma|\cos\delta\theta<|2p-1|
 \iff\text{both posteriors of }P_0\text{ lie strictly on one side of }\tfrac12.
 \label{eq:pareto-condC}
\end{equation}
In the normal form, the last condition says that both chord endpoints of
$P_0$ lie strictly on one side of $k$. Likewise $\Pe(P_0)=\Pe^{0}$ holds
exactly when both posteriors lie weakly on one side, that is, when
$\delta\theta\geq\omega_c$. Condition~\eqref{eq:pareto-condC} is therefore
$\Pe(P_0)=\Pe^{0}$ with the boundary $\delta\theta=\omega_c$ excluded. For
$p=1/2$ one has $\omega_c=\pi/2$, and the condition never holds. Gaps are
not excluded there. For $p=1/2$, Bloch radii $1-10^{-6}$ and $0.99$, and
Bloch vectors $0.02$ rad apart, the best orthogonal basis still switches,
and $\Delta^*=3.4\times10^{-4}$ nat ($19\%$ of the small
$\Iacc=1.8\times10^{-3}$ nat; the Helstrom and information-optimal errors
differ here by at most $\Pe^0-\Pe^{\mathrm H}\approx0.006$, below the $0.01$
required in the directed search of Sec.~\ref{app:pareto-size}), consistent
with the condition being sufficient only.
Straddling alone excludes commuting
states from the condition, because their information-optimal basis is the
Helstrom basis, with error $\Pe^{\mathrm H}<\Pe^0$.

\begin{lemma}[In-plane bases]\label{lem:pareto-plane}
Let $x_\tau>0$, so that $\bm r_0,\bm r_1$ span a plane $\Pi\ni\bm\gamma$.
For the two given signal states, the outcome distributions of every PVM
with axis outside $\Pi$ can be obtained by postprocessing those of a PVM
with axis in $\Pi$. Hence $\VPVM(\lambda)$ is attained in $\Pi$, and the
information-optimal PVM has its axis in $\Pi$.
\end{lemma}

\par\textit{Proof}---
Noncommuting states have $\bm r_0\times\bm r_1\neq0$. Write
$\bm n=\cos\psi\,\bm m+\sin\psi\,\bm e$, with $\bm m\in\Pi$ a unit vector,
$\bm e\perp\Pi$ and $\cos\psi\geq0$. Then
$\Tr\rho_i\frac12(I\pm\bm n\cdot\bm\sigma)=\frac{1\pm\cos\psi}2\Tr\rho_i\frac12(I+\bm m\cdot\bm\sigma)
+\frac{1\mp\cos\psi}2\Tr\rho_i\frac12(I-\bm m\cdot\bm\sigma)$. For this signal ensemble, these are the statistics of the PVM
along $\bm m$ followed by a binary symmetric channel, which cannot raise
$I$ or lower $\Pe$. If the information-optimal PVM were off $\Pi$, the PVM
along $\bm m$ would also be optimal, contradicting
Lemma~\ref{lem:pareto-unique}.

\begin{theorem}[A forced switch]\label{thm:pareto-forced}
Let $x_\tau>0$, let the spectrum straddle $1/2$, and assume
condition~\eqref{eq:pareto-condC}. Parametrize the axes in $\Pi$ by
$\theta\in\mathbb{RP}^1$, and let $\Theta^*(\lambda)$ be the set of maximizers
of $I(\theta)-\lambda\Pe(\theta)$. No continuous map $\theta(\lambda)$ on
$[0,\infty)$ satisfies $\theta(\lambda)\in\Theta^*(\lambda)$ for all $\lambda$.
For every choice of optimal in-plane bases there is a $\lambda_J>0$ at which
$\Theta^*(\lambda_J)$ has at least two elements, and there
$\Delta_{\lambda_J}>0$.
\end{theorem}

\par\textit{Proof}---
The objective is jointly continuous on $[0,\infty)\times\mathbb{RP}^1$, and
$\mathbb{RP}^1$ is compact, so by Berge's theorem $\Theta^*$ has nonempty
compact values and a closed graph. By Lemmas~\ref{lem:pareto-unique}
and~\ref{lem:pareto-plane}, $\Theta^*(0)=\{\theta_0\}$, and $I(\theta)<I(\theta_0)$
for $\theta\neq\theta_0$. The Helstrom axis $\theta_H$ is the unique
minimizer of $\Pe$ on $\mathbb{RP}^1$. For $\theta(\lambda)\in\Theta^*(\lambda)$,
$\Pe(\theta(\lambda))-\Pe(\theta_H)\leq[I(\theta(\lambda))-I(\theta_H)]/\lambda\leq h(p)/\lambda$,
so $\theta(\lambda)\to\theta_H$. Let $W$ be the closed arc of axes with
$\omega\geq\omega_c$. Then $\theta_0$ lies in the interior of $W$ by
Eq.~\eqref{eq:pareto-condC}, and $\theta_H\notin W$ because $\omega_c>0$. On
$W$, $\Pe=\Pe^{0}=\Pe(\theta_0)$ by Lemma~\ref{lem:pareto-pvm-error}, so every
$\theta\in W\setminus\{\theta_0\}$ satisfies
$I(\theta)-\lambda\Pe(\theta)<I(\theta_0)-\lambda\Pe(\theta_0)$ for all
$\lambda\geq0$ and is never optimal. A continuous selection would have a
connected image containing $\theta_0\in\operatorname{int}W$ and points outside
$W$, and would therefore meet $\partial W\subset W\setminus\{\theta_0\}$:
a contradiction. Any selection is therefore discontinuous at some
$\lambda_J$. Compactness and the closed graph give
$\theta'\in\Theta^*(\lambda_J)$ with $\theta'\neq\theta(\lambda_J)$, and
$\lambda_J>0$ because $\Theta^*(0)$ is a single point. Both tied PVMs attain
$\VPVM(\lambda_J)$ by Lemma~\ref{lem:pareto-plane}, and
Theorem~\ref{thm:pareto-tie} gives $\Delta_{\lambda_J}>0$. Near
$\theta_0$ the error is constant, so $\theta_0$ remains a local maximum for
every $\lambda$: the informative basis is locked, and it can only be left
abruptly. In particular, let $\lambda_1$ be the largest weight at which
$\theta_0$ is optimal; it is finite because optimal bases tend to
$\theta_H\notin W$, and it is attained by the closed graph. Optimal bases for
$\lambda\downarrow\lambda_1$ lie outside $W$, since bases in
$W\setminus\{\theta_0\}$ are never optimal, and their limit points are optimal
at $\lambda_1$, hence also outside $W$. The first switch is therefore between
the locked basis, with $\Pe=\Pe^0$, and a basis outside $W$, which has
$\Pe<\Pe^0$ by Lemma~\ref{lem:pareto-pvm-error}: the jump in $\Pe$ is
strictly positive.

\begin{corollary}[Onset of the gap]\label{cor:pareto-onset}
Under the hypotheses of Theorem~\ref{thm:pareto-forced}, let the contacts
$z_1<z_2$ of $P_0$ lie below $k$; otherwise replace $z-k$ by $k-z$ and
$(k,1]$ by $[-1,k)$. Let $D_0$ be the dual function of $P_0$ at
$\lambda=0$. Put
\begin{equation}
 \lambda_a=\min_{z\in(k,1]}\frac{D_0(z)}{2d(z-k)}.
 \label{eq:pareto-onset}
\end{equation}
Then $0<\lambda_a<\infty$, and the minimum is attained at a point
$z_3\in(k,1)$. The PVM $P_0$ is optimal among all POVMs exactly for
$0\leq\lambda\leq\lambda_a$. Moreover, $\Delta_\lambda>0$ on
$(\lambda_a,\lambda_a+\epsilon)$ for some $\epsilon>0$, and
$\lambda_a<\lambda_J$ for every $\lambda_J$ of Theorem~\ref{thm:pareto-forced}.
\end{corollary}

\par\textit{Proof}---
The dual certifying a chord is unique. Two duals with contact conditions
$D=D'=0$ at $z_1,z_2$ differ by
$L=\delta\alpha+\delta\beta z-\delta\kappa s$ with double zeros at $z_1$ and $z_2$. That
is impossible for strictly convex or strictly concave $L$, and for affine $L$
it forces $L\equiv0$. Since $g_\lambda-g_0=\lambda d|z-k|$,
\begin{equation}
 D_\lambda(z)=D_0(z)-2\lambda d\,(z-k)_+
\end{equation}
is the dual function of $(\alpha_0+\lambda dk,\beta_0-\lambda d,\kappa_0)$ at
weight $\lambda$. It satisfies the contact conditions at $z_1,z_2<k$, so it
is the only possible certificate for $P_0$, and it is feasible exactly for
$\lambda\leq\lambda_a$. The function $D_0$ is positive on $[k,1]$ by
Lemma~\ref{lem:pareto-contacts}, so the ratio tends to $+\infty$ as
$z\downarrow k$; it is continuous on $(k,1]$, which gives $0<\lambda_a<\infty$. A minimizer at $z=1$ would be a zero of
$D_{\lambda_a}$ at a pole, which Lemma~\ref{lem:pareto-contacts} excludes.
For $\lambda>\lambda_a$ suppose $\Delta_\lambda=0$ arbitrarily close to
$\lambda_a$. At each such $\lambda$ the unique optimum $\mu^*_\lambda$ is
attained by a PVM and therefore has exactly two atoms. As
$\lambda\downarrow\lambda_a$ it converges weakly to $P_0$ (Berge), and
two-atom measures that converge weakly to a two-atom limit with positive
weights have converging atoms. So $\mu^*_\lambda$ has one atom in each of
two disjoint neighbourhoods of $z_1$ and $z_2$, both below $k$. A chord below $k$ that is optimal at $\lambda$ is optimal
at every smaller weight, because adding $2(\lambda-\lambda')d(z-k)_+\geq0$ to
its dual keeps it feasible. It would then be the information optimum $P_0$,
which is not optimal for $\lambda>\lambda_a$: a contradiction. Finally,
$\Delta=0$ on $[0,\lambda_a]$, so Theorem~\ref{thm:pareto-tie} excludes ties
there. This is the construction of Sec.~\ref{app:pareto-construct}, with
$\tau$ given and the onset weight explicit.

Numerical checks of the contact geometry, onset, and basis switches are
collected in Sec.~\ref{app:pareto-size}.

\prlappendix{Numerical examples and size of the three-outcome advantage}
\label{app:pareto-size}
\label{app:pareto-numerics}
\textit{Methods and certification.} The scans locate candidate advantages;
we then certify each of the 1100 positive random pairs at its selected
weight. An explicit feasible three-outcome POVM gives a lower bound, and
an outward-rounded interval search bounds every projective basis from
above. Each angular cell uses a Taylor bound with an interval enclosure
of the second derivative, rather than a sampled curvature maximum.
A separate interval search bounds the global dual function for the POVM
problem. Its upper bound and the explicit POVM value differ by less than
$1.2\times10^{-13}$ nat in all 1100 cases. The smallest certified lower
bound on the advantage exceeds $4.2\times10^{-9}$ nat. Inputs are the
stored floating-point Bloch vectors, treated as exact dyadic numbers;
interval arithmetic uses 50 decimal digits. Independent searches over
complex $K$-outcome POVMs ($K=2,3,4$, without a planar restriction)
provide an additional numerical cross-check.

\prlsubheading{Dual-slack scan and interval location}
At each weight we optimize the in-plane PVM and construct its tangent
dual function. We search the remaining angles for a dual violation and
check the curvature at the two contacts. A logarithmic weight grid with
20 points per decade locates candidate gap intervals; local refinement
and bisection in $\log_{10}\lambda$ locate their endpoints.
This floating-point scan can miss narrow intervals, so positive witnesses
are certified separately by the interval bounds above.
For the complementary example below, an initial exploratory scan used
321 support weights, 1440 posterior-circle angles, and an 8192-point PVM
mesh, with gap threshold $10^{-7}$ nat.

In the interval PVM search, $F_\lambda$ is the maximum of three smooth
branches corresponding to the possible guessing decisions. We bound each branch separately
by its interval Taylor enclosure and then take the maximum, which also
covers angular cells containing a decision threshold.

\prlsubheading{Representative ensembles}
The main-text ensemble of Figs.~\ref{fig:schematic} and
\ref{fig:pareto-gap}(a) has $c=0.30$, $d=0.29$,
$\tau=(I+0.79\sigma_x+0.16\sigma_z)/2$, and $p=0.3464$.
Its information-optimal posteriors are $0.055$ and $0.441$, both favoring
$\rho_1$, so its error is $\Pe=p$. With $\omega_c=0.680$ rad and margin
$\delta\theta-\omega_c=0.56$ rad, it satisfies
condition~\eqref{eq:pareto-condC}. The gap occupies
$(\lambda_a,\lambda_b)=(0.288,0.832)$ and reaches
$\Delta^*=2.308\times10^{-3}$ nat ($3.1\%$ of $\Iacc=0.0753$ nat)
at the basis switch $\lambda_J=\lambda^*=0.510$.
At $\lambda=0.6$, $\Delta_\lambda=1.1929\times10^{-3}$ nat,
confirmed by three methods.
At the switch, two PVMs tie: the maximum-information basis
($\Pe=0.3464$, $I=0.0753$ nat) and a basis with $\Pe=0.3124$ and
$I=0.0579$ nat. Between these points the fixed-error frontier of PVMs
and their random mixtures is their straight mixing segment
[Corollary~\ref{cor:pareto-tie}(d)], while the POVM frontier lies
strictly above its interior. The Helstrom measurement has $\Pe=0.3024$
and $I=0.0460$ nat. The fixed-error gain is positive for
$0.307<\varepsilon<0.346$, with maximum $2.31\times10^{-3}$ nat at
$\varepsilon=0.327$.

Numerically, the outcome count changes by vanishing weights at both ends
of this window. The third outcome, which weakly favors state 0
($\eta\approx0.53$), enters at $\lambda_a$ with weight
$\approx1.66(\lambda-\lambda_a)$; the outcome that started at
$\eta=0.055$ leaves at $\lambda_b$ (where $\eta\approx0.20$), with weight
$\approx50(\lambda_b-\lambda)$. The optimum is numerically not $C^2$:
finite differences give jumps of $\dd\Pe/\dd\lambda$ from $0$ to $-0.105$
at $\lambda_a$ and from $-0.052$ to $-0.009$ at $\lambda_b$, while
$\Pe$ is continuous. At $\lambda_J$ the best PVM axis turns by
$0.918$ rad and its error drops from $0.346$ to $0.312$.

A complementary ensemble has a three-outcome advantage without an abrupt
basis switch. In the Bloch frame, write $\rho_0=(I+u_0\sigma_z)/2$ and
$\rho_1=[I+u_1(\cos\theta\,\sigma_z+\sin\theta\,\sigma_x)]/2$, with
$p=0.4072397056684426$, $u_0=0.9170949735598445$,
$u_1=0.5130827427120952$, and $\theta=0.7195921216129726$ rad.
Its information-optimal posteriors are $0.097$ and $0.501$, so it does
not satisfy condition~\eqref{eq:pareto-condC}.
It has $\Iacc=0.0698$ nat and $\Delta^*=5.78\times10^{-4}$ nat
($0.83\%$ of $\Iacc$) at $\lambda^*=0.639$; the gap occupies
$0.029<\lambda<1.66$, and the best PVM basis rotates continuously.
At $\lambda^*$ the three posteriors are $0.458$, $0.120$, and $0.515$:
two give the same guess, and none sits at $1/2$.

A large-gap ensemble away from degenerate corners has $p=0.67084$,
Bloch radii $0.5045$ and $0.9494$, and angle $1.145$ rad between the
Bloch vectors. It has $\Iacc=0.1084$ nat and
$\Delta^*=7.03\times10^{-3}$ nat ($6.5\%$ of $\Iacc$) at
$\lambda^*=6.8$, where the best PVM basis jumps.

\textit{Random pairs.} We drew 20000 pairs (seeds 20261022 and 20261023)
with priors drawn uniformly from $(0,1)$ and Bloch vectors uniform in the
ball, the Hilbert--Schmidt measure. Their extreme posteriors straddle $1/2$ in
$63.7\%$ of pairs; a certified gap occurs in $8.6\%$ of these and in none of
the 7252 one-sided pairs, as the spectral criterion requires; over all
pairs, a certified gap occurs in $5.5\%$ (1100 pairs). As threshold
variants, $\Delta^*>10^{-7}$ nat holds in $5.45\%$, $\Delta^*>10^{-3}$ nat
in $1.40\%$, and $\Delta^*/\Iacc>1\%$ in $1.75\%$. Among gapped pairs the median
of $\Delta^*$ is $3.4\times10^{-4}$ nat, the 90th percentile
$2.2\times10^{-3}$ nat, and the maximum $0.0117$ nat; the largest
$\Delta^*/\Iacc$ is $13.7\%$. Table~\ref{tab:gap-angle} gives the gapped fraction of the straddling pairs
against the angle $\delta\theta$ between the Bloch axes of the
information-optimal and Helstrom measurements, and Table~\ref{tab:gap-bins}
the bin counts of Fig.~\ref{fig:gap-stats}, which uses the margin
$\delta\theta-\omega_c$ of Sec.~\ref{app:pareto-switch}.

\begin{figure}[tb]
 \centering
 \includegraphics[width=0.5\textwidth]{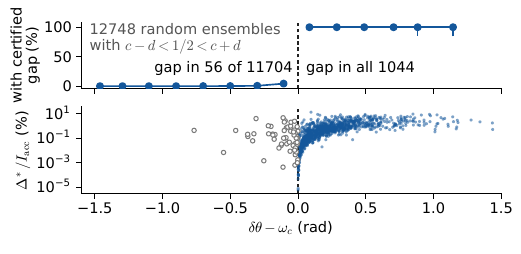}
 \caption{Gaps against the margin of Theorem~\ref{thm:pareto-forced}. The 12748
 random pairs (priors uniform on $(0,1)$, Bloch vectors uniform in the ball)
 whose extreme posteriors straddle $1/2$, against $\delta\theta-\omega_c$.
 Top: percentage with a certified gap (bin counts in
 Table~\ref{tab:gap-bins}); all 1044 pairs above the threshold have a gap,
 and 56 of the 11704 below it do. Bottom: $\Delta^*/\Iacc$ for every gapped
 pair (filled: $\delta\theta>\omega_c$; open: $\delta\theta<\omega_c$). Gaps
 grow from zero above the threshold (median $0.53\%$ of $\Iacc$) and are
 rarer and smaller below it (median $0.14\%$).}
 \label{fig:gap-stats}
\end{figure}

\begin{table}[htbp]
\caption{Fraction of the 12748 straddling random pairs with a certified gap,
in coarse bins of $\delta\theta$; $n$ is the number of pairs per bin.}
\label{tab:gap-angle}
\begin{ruledtabular}
\begin{tabular}{lcccccc}
$\delta\theta$ (rad) & $<0.05$ & $0.05$--$0.1$ & $0.1$--$0.2$ & $0.2$--$0.4$ & $0.4$--$0.8$ & $>0.8$\\
$n$ & 3060 & 2225 & 2952 & 2768 & 1455 & 288\\
with a gap & $0\%$ & $0.4\%$ & $1.9\%$ & $8.1\%$ & $37\%$ & $92\%$\\
\end{tabular}
\end{ruledtabular}
\end{table}

\begin{table}[htbp]
\caption{Bin counts for Fig.~\ref{fig:gap-stats}: straddling random
pairs per bin of the margin $\delta\theta-\omega_c$ and the number with a
certified gap; the totals are 12748 and 1100.}
\label{tab:gap-bins}
\begin{ruledtabular}
\begin{tabular}{lcccccccc}
$\delta\theta-\omega_c$ (rad) & $[-1.6,-1.4)$ & $[-1.4,-1.2)$ & $[-1.2,-1.0)$ & $[-1.0,-0.8)$ & $[-0.8,-0.6)$ & $[-0.6,-0.4)$ & $[-0.4,-0.2)$ & $[-0.2,0)$\\
$n$ & 808 & 1687 & 1784 & 1866 & 1627 & 1528 & 1355 & 1049\\
gapped & 0 & 0 & 0 & 0 & 1 & 2 & 8 & 45\\
\colrule
$\delta\theta-\omega_c$ (rad) & $[0,0.2)$ & $[0.2,0.4)$ & $[0.4,0.6)$ & $[0.6,0.8)$ & $[0.8,1.0)$ & $[1.0,1.6)$ & & \\
$n$ & 526 & 297 & 125 & 52 & 23 & 21 & & \\
gapped & 526 & 297 & 125 & 52 & 23 & 21 & & \\
\end{tabular}
\end{ruledtabular}
\end{table}

\prlsubheading{Checks of contact geometry and basis switches}
These checks use a separate implementation in the normal form,
independent of the initial ensemble-scan code.
\begin{itemize}
 \item \textit{Contact count}. We tested 300
 optimal duals from the random-pair survey above: 150 at
 the gap maximum $\lambda^*$ of gapped pairs and 150 at random
 $\lambda\in[10^{-2},10^2]$, using seed 20261028. They have 148 two-contact and 152
 three-contact duals. All contacts satisfy $R\leq\kappa$, none lies at
 $k$ or at a pole, and every three-contact certificate splits two plus one
 about $k$, with $R>\kappa$ between the pair. With a zero tolerance of
 $10^{-9}$, four two-contact duals showed only one contact. With tolerance
 $10^{-7}$ all 300 have the counts above.
 Two instances, at $\lambda\approx19.8$ and $2.046$, have numerical
 primal--dual residuals of $3.7$--$3.8\times10^{-9}$ nat.
 \item \textit{Switches}. The single-PVM performance
 curve is sampled at 16384 chord angles. A switch is an edge of the
 upper-left convex hull of the points $(\Pe,I)$ that skips more than three
 grid steps. Numerical root finding locates the tie $\lambda_J$, and $\Delta_{\lambda_J}$ is
 bracketed by the moment linear program (LP). Among the 20000 random pairs (12748 straddling), the
 best basis switches in 1044. Each has $\Delta_{\lambda_J}>0$: the smallest is
 $4.3\times10^{-9}$ nat, the brackets are at most $3\times10^{-13}$ nat, and
 the optimum has three atoms. Condition~\eqref{eq:pareto-condC} holds in
 exactly these 1044 pairs. One switching pair whose gap
 interval was too narrow for the 20-per-decade weight grid has
 $\Delta_{\lambda_J}=7.5\times10^{-7}$ nat. Hence 1100 pairs have a gap:
 1044 that switch and 56 that do not, none of the 56 satisfying
 condition~\eqref{eq:pareto-condC}. In the $10^4$ pairs biased toward high purity
 (8480 straddling), 439 switch, all with $\Delta_{\lambda_J}>0$ (smallest
 $4.5\times10^{-10}$ nat). Of these, 435 satisfy the condition, and every pair
 satisfying it switches. Four switch without it; they have nearly pure,
 nearly parallel Bloch vectors and $0.49<p<0.55$. The condition is therefore
 sufficient but not necessary.
 \item \textit{Angle bins}. In the bins of
 Table~\ref{tab:gap-angle}, the gapped fraction follows the fraction
 satisfying the condition: $0.374$ against $0.355$ for
 $0.4\leq\delta\theta<0.8$ rad, and $0.924$ against $0.837$ above $0.8$ rad.
 \item \textit{Onset}. For 200 pairs satisfying the
 condition, Eq.~\eqref{eq:pareto-onset} matches the start of the gap
 interval found by the dual-slack scan above, with
 median relative deviation $5.1\times10^{-8}$. The three largest
 deviations occur at $\lambda_a\approx10^{-3}$: two scan edges are late,
 while the third is slightly early. Direct moment-LP checks at
 $0.9\lambda_a$ and $1.1\lambda_a$ find no resolved advantage below
 the predicted onset and a positive advantage above it in all three cases.
 These floating-point checks support Eq.~\eqref{eq:pareto-onset}; the
 dual-slack scan alone does not resolve such narrow intervals reliably.
 \item \textit{Anatomy of gap intervals}. The
 optimum just inside each end of the interval identifies the outcome that
 appears and the outcome that leaves. In every interval satisfying the
 condition (80 of 80 in a subsample of the random pairs, 431 of 431 in the
 purity-biased set), the new outcome is the lone outcome across $\eta=1/2$,
 outside the informative chord. In every other interval (56 of 56 and 28 of
 28), it appears between the two orthogonal outcomes, just across $\eta=1/2$
 from the weaker one. Two types occur. In a \emph{bubble} (68 intervals) the
 new outcome disappears again; none of these has a switch. In a
 \emph{handover} (16 intervals) the original same-guess outcome leaves
 instead; four of these switch, and they are the four switches without the
 condition. Four condition-satisfying purity-biased pairs were excluded
 from this endpoint classification: three had gap intervals too narrow
 for the scan to resolve, and one returned two atoms at its upper end.
 Their advantages at the basis switch were separately verified by
 outward-rounded interval bounds. Among the random pairs, the median
 $\Delta^*$ is $3.7\times10^{-4}$ nat for the switching pairs and
 $9.7\times10^{-5}$ nat for the 56 others ($7.1\times10^{-5}$ for
 bubbles).
 \item \textit{Structure at the switch}. We checked 150 switches satisfying the
 condition. In all of them, the tied informative basis has both posteriors on
 one side of $1/2$ and the other tied basis straddles. The three-outcome
 optimum has two outcomes near those of the informative basis (largest
 offset $0.22$ in $z$) and a third across $1/2$, near the outcome of the
 other basis that favors the less likely state. This dissenting outcome,
 whose guess differs from the other two, is the least confident of the
 three in all 150 cases, with median $|\eta-\tfrac12|=0.014$; its weight is
 not small (median $0.13$, up to $0.45$). An independent check on 50
 ensembles finds the same structure, with kept outcomes shifted by a median
 of $0.07$ in $z$ (up to $0.23$), median $|\eta-\tfrac12|=0.009$ for the
 dissenting outcome, and dissenting weights of median $0.18$ (up to $0.40$).
 \end{itemize}
Every classified interval follows one of these entry-and-exit patterns.
The connectivity question is discussed below.

\textit{Interval structure and open extensions.} No ensemble with two or more
intervals of weights with a gap was found among the 20000 random pairs,
$10^4$ pairs biased toward high purity, and a family near the pure-state
corner. The number and connectivity of such intervals remain open.
Extending the Pareto analysis to rank-deficient states is another direction.

\textit{Largest gaps.} A directed search (differential evolution) under the
constraint that the error window between the Helstrom and
information-optimal measurements is at least $0.01$ reaches
$\Delta^*=0.0133$ nat, $17\%$ of $\Iacc$; its optimizer sits at the purity
bound of the search, Bloch radius $1-10^{-6}$ for one state. Nearly pure
states raise the value but are not needed for large gaps: with the other
parameters fixed, Bloch radius
$0.99$ or $0.95$ gives $16.6\%$ or $14.9\%$ (error windows $0.009$ and
$0.003$), and a local search with both radii at most $0.99$ and the window
kept at least $0.01$ reaches $\Delta^*=0.0114$ nat ($11\%$ of $\Iacc$).
Without the
window constraint the gap approaches $0.0157$ nat only in a degenerate
corner, with one state pure and the error window closing.
Relative to $\Iacc$, gaps can be larger where $\Iacc$ is small.
A search for a large $\Delta^*/\Iacc$ under the same window gives $34\%$
($\Delta^*=5.78\times10^{-3}$ nat, $\Iacc=0.0171$ nat; certified, recomputed
independently), with the condition of Theorem~\ref{thm:pareto-forced} met and the
optimizer on both the window and the purity bound; at Bloch radius $0.99$ for
that state, the other parameters fixed, the ratio is $29\%$, with window $0.009$. These are certified values, not
established maxima. Without the window the ratio reaches $64\%$, at an optimizer in
a degenerate corner with $\Iacc\approx0.003$ nat and nearly parallel states.

For a reproducible point strictly inside the window constraint, take
\begin{align*}
 p&=0.5128178885, & u_0&=0.9051506613,\\
 u_1&=0.999999, & \theta&=0.0829064461
\end{align*}
in the Bloch-frame parametrization introduced for the complementary
ensemble above.
At $\lambda=1.46495$, outward-rounded interval bounds give
$V_{\rm POVM}-V_{\rm PVM}>0.00578223$ nat and a relative gain above
$33.7628\%$. They also verify that every information-optimal PVM makes
one constant guess, so its error window equals
$(\lVert p\bm r_0-(1-p)\bm r_1\rVert-|2p-1|)/2>0.01$.
These are bounds for this explicit ensemble, not a claim of a global
maximum over ensembles.

\textit{Degenerate limits.} No gap occurs in 304 ensembles on the collapse
surface $p^2\det\rho_0=(1-p)^2\det\rho_1$, including unequal priors with
nearly pure states, where the information-optimal and Helstrom bases agree
to $8\times10^{-8}$ rad. The construction of Sec.~\ref{app:pareto-construct} was checked at 19
points $(c,d)$, with every certificate nonnegative; the gap at the
constructed average state is tiny when $k=(1/2-c)/d$ is small, about
$10^{-7}$ nat at $k=0.1$. Targeted searches examined 25 straddling points and five controls.
Six straddling points required a finer search near the nearly pure boundary;
the others were resolved by the initial scan. At $d=0.2$, six locally
optimized witnesses with $c=0.49,0.495,0.498,0.499,0.501,0.505$ have gaps
approximately $1.85\delta_c^4$ nat, with $\delta_c=|c-1/2|$. This is a
numerical scaling observation, not an asymptotic theorem or a global
optimization claim. Outward-rounded bounds certify positive gaps even
at $c=0.499,d=0.2$ (above $1.85\times10^{-12}$ nat) and
$c=0.4995,d=0.2966$ (above $8.30\times10^{-14}$ nat). In the latter case
the global PVM enclosure has width $4.84\times10^{-16}$ nat: it suffices
to establish positivity although the requested $10^{-17}$ precision was
not reached within the subdivision budget.

\prlappendix{Operational meaning and a physical setting}
\label{app:pareto-physical}
\textit{The readout task.} Three outcomes also arise in discrimination at a
fixed inconclusive rate
\cite{chefles1998strategies,fiurasek2003optimal,bagan2012optimal}, under an
error margin \cite{hayashi2008state,sugimoto2009discrimination}, or for
maximum confidence \cite{croke2006maximum,herzog2009discrimination}.
Bounded-error projective strategies are studied in Ref.~\cite{touzel2007optimal}.
For two pure states, the losses of Ref.~\cite{katz2024optimal} select either
the Helstrom projection or, for priors that are not too skewed, the
unambiguous measurement. Here every outcome leads to a guess. The third
outcome retains confidence information within a guess, as shown in
Fig.~\ref{fig:schematic}(b). In the 150 sampled switches described in
Sec.~\ref{app:pareto-size} under ``Structure at the switch,'' the outcome favoring the less likely state
has median $|\eta-1/2|=0.014$.

\textit{Rate at a fixed per-symbol error.} A message is encoded into a
codeword $x^n\in\{0,1\}^n$ of constant composition $p$, that is, with a
fraction $p$ of zeros, and use $j$ delivers $\rho_{x_j}$. The receiver
applies the same measurement $M$ to every system separately, which induces
the classical channel $W_M(y|x)=\Tr\rho_xM_y$. It must (i) output an
immediate decision $\hat x_j=\delta(y_j)$ in every use and (ii) pass $y^n$ to
a decoder of the message. Because a fraction $p$ of the uses carry $\rho_0$,
the decision error averaged over the uses is the Bayes error of $M$ at prior
$p$, which the best rule reduces to $\Pe(M)$; we require $\Pe(M)\leq\epsilon$.
The supremum of asymptotically achievable rates with constant-composition
codes approaching type $p$ is $I(p,W_M)=I(\mathcal E,M)$
\cite{csiszar2011information}. Over generalized measurements the best rate
compatible with (i) is therefore $I_{\rm POVM}(\epsilon)$ of
Sec.~\ref{app:extension-overview}. Orthogonal measurements may also be chosen at
random in each use, independently of the message (i.i.d.\ and unknown to the
encoder), with the choice known to the decision rule and the decoder. The
effective channel is then a single memoryless channel whose output includes
the basis label; it meets (i) on average over the random choice. Its supremum
of achievable rates is $\sum_bw_bI(\mathcal E,P_b)$, and the same converse bounds
every such strategy by $I_{\rm PVM}(\epsilon)$, the concave hull that
random switching attains. Hence
$\Delta I(\epsilon)=I_{\rm POVM}(\epsilon)-I_{\rm PVM}(\epsilon)$ is a rate
gain at a fixed per-symbol error, and $\max_\epsilon\Delta I=\Delta^*$. A
deterministic schedule of bases known to the encoder is a different problem,
not treated here.

\textit{Scope of the coding interpretation.} The input distribution is fixed, so this is not the
classical capacity of the channel $x\mapsto\rho_x$. Collective and adaptive measurements over many uses are outside this
receiver model; unrestricted collective decoding can attain the Holevo rate.
The constraint fixes the Bayes error
$p\Pr(\hat x\ne0|0)+(1-p)\Pr(\hat x\ne1|1)$, not the two error types
separately.

\textit{Source-coding meaning.} For a memoryless source emitting
$X_j\in\{0,1\}$ with $\Pr(X_j=0)=p$, observed through the same $M$ in every
shot with per-shot decision error $\Pe(M)\leq\epsilon$, the Slepian--Wolf
theorem \cite{slepian1973noiseless} gives $h(p)-I(\mathcal E,M)$ nats per
shot as the infimum asymptotically achievable rate of additional data from
an encoder observing $X^n$, when a decoder holding $Y^n$ must reconstruct
$X^n$ with vanishing block error. At a fixed per-shot decision error, the
three-outcome receiver lowers this rate by $\Delta I(\epsilon)$.

\textit{A physical setting.} A qubit is prepared in $|0\rangle$ with prior
$p$ or in $\cos(\theta/2)|0\rangle+\sin(\theta/2)|1\rangle$ with prior $1-p$
and sent through a generalized amplitude-damping channel with loss $\Gamma$
and thermal population $n_{\rm th}$ (the population of $|1\rangle$ in the
environment's state), as in on-off keying over a lossy line or a
$T_1$-limited pointer. The channel maps Bloch vectors as
$(x,y,z)\mapsto(\sqrt{1-\Gamma}\,x,\sqrt{1-\Gamma}\,y,(1-\Gamma)z+\Gamma z_\infty)$,
with $z_\infty=1-2n_{\rm th}$. The error window below is the relative gap
$(\Pe^0-\Pe^{\mathrm H})/\Pe^0$ between the prior and Helstrom errors.

The channel scans cover 12600 preparations at each of $\Gamma=0.3$ and
$0.5$. For 102 channel/prior/window configurations, we use differential
evolution with seed 7, Sobol initialization, population multiplier 30, and
80 generations, then refine the measurement optimum. Positive advantages
are found in 92 cases; all pass the independent outward-rounded PVM-bound
check. The other ten searches are inconclusive about existence.
Figure~\ref{fig:physical}(b) reports the best values found within this budget.

At $p=0.2$, $\Gamma=0.3$, $\theta=2.1648$ rad ($0.689\pi$) and
$n_{\rm th}=0.01$, the states have Bloch vectors $(0,0,0.994)$ and
$(0.6933,0,-0.0978)$, with $\Iacc=0.1512$ nat, $\Pe^0=0.200$,
$\Pe^{\mathrm H}=0.190$ and an error window of $5\%$. The information-optimal
basis leaves posteriors $0.43$ and $0.02$ for $|0\rangle$, so both of its
outcomes favor the common state; the margin is
$\delta\theta-\omega_c=0.699-0.255=0.44$ rad, and
condition~\eqref{eq:pareto-condC} holds. The gap opens at $\lambda_a=2.26$, the
best basis switches at $\lambda_J=3.46$, and
$\Delta^*=\Delta_{\lambda_J}=3.12\times10^{-3}$ nat ($2.07\%$ of $\Iacc$);
$\Delta I(\epsilon)>0$ for $0.1907<\epsilon<0.2000$.
Table~\ref{tab:phys} lists the three-outcome optimum and the two tied PVMs at
$\lambda_J$; the third outcome is the one that favors $|0\rangle$, with
low confidence. For comparison, a
dephased Ramsey sensor (probe $|+\rangle$, dephasing coherence at most
$0.99$, a rare event that rotates the phase by $\varphi$, equal Bloch radii)
yields best-found gains of $0.37\%$ of $\Iacc$ at $p=0.1$ and $0.81\%$ at $p=0.2$ for
error windows of at least $5\%$ (Fig.~\ref{fig:physical}). The lossy regime
occupies a band of preparations [Fig.~\ref{fig:physical}(a)]: with $p=0.2$
and $\Gamma=0.3$, gaps occur at $\theta=0.55\pi$, $0.62\pi$, and
$0.689\pi$ and not at $0.75\pi$ or $0.85\pi$, where the extreme posteriors
no longer straddle $1/2$. At $\theta=0.689\pi$, $\Gamma=0.1$ or $0.5$
gives no gap; at $\Gamma=0.5$ the gaps move to $\theta\lesssim0.48\pi$.
For rarer states, the largest gains found in these searches fall roughly as $p^2$, to about $0.5p^2\,\Iacc$ for error windows of
at least $5\%$, and below $1\%$ of $\Iacc$ at $p=0.1$ even without that
restriction [Fig.~\ref{fig:physical}(b)].

\begin{table}[htbp]
\caption{Measurements at the switch $\lambda_J=3.46$ of the operating point:
outcome probabilities $q_y$, effect weights $\Tr M_y/2$, posteriors
$\Pr(|0\rangle\,|\,y)$, decisions (c: common state; r: $|0\rangle$),
information and error.}
\label{tab:phys}
\begin{ruledtabular}
\begin{tabular}{lcccccc}
Measurement & $q_y$ & $\Tr M_y/2$ & $\Pr(|0\rangle|y)$ & decisions & $I$ (nat) & $\Pe$\\
\colrule
Three-outcome optimum & $0.263,0.159,0.578$ & $0.254,0.271,0.475$ & $0.38,0.52,0.03$ & c, r, c & $0.1375$ & $0.1951$\\
Informative PVM (tied) & $0.445,0.555$ & $0.5,0.5$ & $0.43,0.02$ & c, c & $0.1512$ & $0.2000$\\
Second PVM (tied) & $0.296,0.704$ & $0.5,0.5$ & $0.51,0.07$ & r, c & $0.1210$ & $0.1913$\\
\end{tabular}
\end{ruledtabular}
\end{table}

\begin{figure}[tb]
 \centering
 \includegraphics[width=0.5\textwidth]{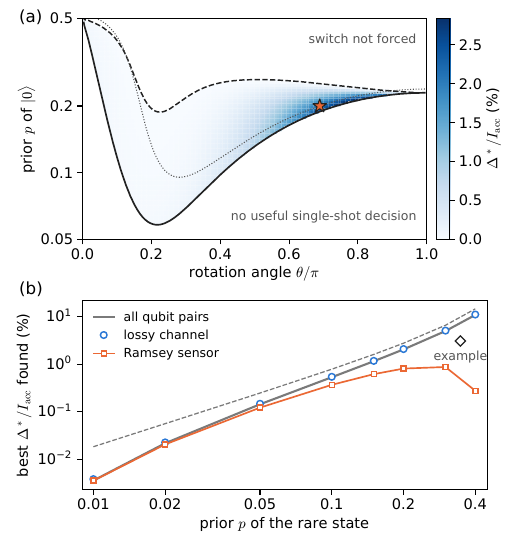}
 \caption{A lossy channel. (a) A qubit is sent in $|0\rangle$ (prior $p$) or
 rotated by $\theta$ towards $|1\rangle$, through amplitude damping with loss
 $\Gamma=0.3$ and thermal population $0.01$. Color: $\Delta^*/\Iacc$ where a
 gap exists. Below the solid line no measurement beats guessing the likelier
 state; above the dashed line the switch is not forced. Above the dotted
 line the Helstrom error is at least $5\%$ below the prior error. Star: the
 operating point of the text. (b) Largest $\Delta^*/\Iacc$ found at fixed $p$, for
 error windows of at least $5\%$, over all qubit pairs with Bloch radii
 $\le0.99$ (gray), the lossy channel with $\Gamma$, $\theta$ and the thermal
 population free (circles), and a dephased Ramsey sensor (squares). The gray
 curve depends on the imposed cutoffs (radius $0.99$ and window $\geq5\%$); the dashed curve shows the best values found without the window
 constraint. Diamond: the ensemble defined in the Letter
 Fig.~\ref{fig:pareto-gap} caption (also shown in
 Fig.~\ref{fig:schematic}).}
 \label{fig:physical}
\end{figure}

\prlappendix{Unequal decision costs and other information scores}
\label{app:pareto-general}
\textit{Setting.} Let deciding ``$1$'' when the state is $\rho_0$ cost $a>0$,
and deciding ``$0$'' when the state is $\rho_1$ cost $b>0$. The Bayes cost of
a measurement is
\begin{equation}
 C(M)=\sum_y\min\bigl(a\Tr AM_y,\,b\Tr BM_y\bigr)
 =\sum_yq_y\min\bigl(a\eta_y,\,b(1-\eta_y)\bigr),
 \label{eq:gen-cost}
\end{equation}
so outcome $y$ calls for ``$0$'' exactly when $\eta_y>t_0=b/(a+b)$. Equal
costs give $C=a\Pe$ and $t_0=1/2$. Replace the entropy by a generalized
entropy $f$, assumed to be $C^2$ and strictly concave on $(0,1)$, with
weight function $w=-f''\geq0$; this is the entropy of a strictly proper
scoring rule~\cite{savage1971elicitation,gneiting2007strictly}. Let $I_f(M)=f(p)-\sum_yq_yf(\eta_y)$, which is the mutual
information for $f=h$. We study $F=I_f-\lambda C$ with $\lambda\geq0$.

\textit{Reduction.} Since
$\min(a\eta,b(1-\eta))=\frac12[a\eta+b(1-\eta)]-\frac{a+b}2|\eta-t_0|$ and
$\sum_yq_y\eta_y=p$, the maximum of $F$ equals
$f(p)-\frac\lambda2[ap+b(1-p)]$ plus the maximum of Eq.~\eqref{eq:scalar},
with
\begin{equation}
 g(z)=-f(c+dz)+\bar\lambda\,|c+dz-t_0|,\qquad
 \bar\lambda=\tfrac{a+b}2\lambda,\qquad k=\frac{t_0-c}d.
 \label{eq:gen-g}
\end{equation}
The rank-one reduction and scalar moment formulation of
Sec.~\ref{app:pareto-scalar} require two convexity properties. The
per-outcome objective $-f(\eta)-\lambda\min(a\eta,b(1-\eta))$ must be convex in
$\eta$, for refinement to rank one; the Bayes loss itself is concave, so its
negative is convex. Saturation requires strict convexity of $g$. Both properties hold, so the outcome count, the dual
Eq.~\eqref{eq:dual} and Lemma~\ref{lem:pareto-dual} carry over unchanged. Off
$z=k$, Eq.~\eqref{eq:pareto-Dpp} becomes $D''=[\kappa-R_f(z)]/s(z)^3$ with
\begin{equation}
 R_f(z)=d^2s(z)^3\,w(c+dz).
 \label{eq:gen-R}
\end{equation}

\textit{Shape condition.} Say that $f$ satisfies (U) for the spectrum
$c\pm d$ if $R_f$ is strictly increasing and then strictly decreasing on
$(-1,1)$. Then $\{R_f>\kappa\}$ is an interval for every $\kappa$.

\begin{lemma}[A sufficient condition]\label{lem:gen-logconcave}
If $W(\eta)=w(\eta)\,[\eta(1-\eta)]^{3/2}$ is log-concave on $(0,1)$, then
$R_f$ is strictly log-concave on $(-1,1)$ for every spectrum
$0<c-d<c+d<1$, and in particular (U) holds.
\end{lemma}

\par\textit{Proof}---
Strict concavity of $f$ forbids $w\equiv0$ on any subinterval, and a
log-concave function vanishes only outside an interval, so $W>0$ on
$(0,1)$. With $\eta=c+dz$, $\alpha=c-d$ and $\beta=c+d$, the identity
$d^2s(z)^2=(\eta-\alpha)(\beta-\eta)$ gives
\begin{equation}
 \log R_f=\log W(\eta)
 +\tfrac32\log\frac{(\eta-\alpha)(\beta-\eta)}{\eta(1-\eta)}+\mathrm{const}.
\end{equation}
The first term is concave by assumption. The second is strictly concave on
$(\alpha,\beta)$: its second derivative,
$\frac32[\eta^{-2}+(1-\eta)^{-2}-(\eta-\alpha)^{-2}-(\beta-\eta)^{-2}]$, is
negative because $0<\alpha<\eta<\beta<1$. Hence $\log R_f$ is strictly concave
in $\eta$, and therefore in $z$. No derivative of $w$ is needed. A strictly
log-concave function that vanishes at $z=\pm1$ is strictly increasing and then
strictly decreasing.

The lemma applies to three standard entropies.
\begin{itemize}
 \item The Shannon entropy: $W=[\eta(1-\eta)]^{1/2}$.
 \item The Brier (Gini) entropy $2\eta(1-\eta)$: $W=4[\eta(1-\eta)]^{3/2}$.
 \item The spherical entropy $1-\sqrt{\eta^2+(1-\eta)^2}$. With
 $u=\eta(1-\eta)$, $\log W=\frac32\log[u/(1-2u)]$. Its second derivative in
 $\eta$ is
 $\frac32\{(1-4u)[4(1-2u)^{-2}-u^{-2}]-2[u^{-1}+2(1-2u)^{-1}]\}<0$, since
 $u\leq1/4$.
\end{itemize}
For the Tsallis entropies~\cite{tsallis1988possible} $[1-\eta^q-(1-\eta)^q]/(q-1)$, the condition holds
numerically for $0.697\lesssim q\lesssim4.30$. For the binary R\'enyi entropies,
which for positive order are concave only for $0<\alpha\leq2$, it holds for
$0.823\lesssim\alpha\leq2$. Outside these ranges $R_f$ numerically has two
maxima for nearly pure spectra, so (U) fails.

\begin{theorem}[General losses]\label{thm:gen}
Let the states be full-rank and noncommuting ($x_\tau>0$), and let $f$
satisfy (U).
(a) For every $\lambda\geq0$ the optimal measure is unique, so the optimal
POVM is unique up to relabeling and splitting effects into proportional
pieces. Its atoms avoid the poles and, for $\lambda>0$, the threshold height
$k$. At $\lambda=0$ it is a PVM.
(b) If two distinct PVMs tie for the best value at $\lambda_J>0$, a
three-outcome measurement is strictly better there.
(c) Write $aA-bB=\frac12(\gamma_0I+\bm\gamma'\cdot\bm\sigma)$, with
$\gamma_0=(a+b)(p-t_0)$ and $\bm\gamma'=ap\bm r_0-b(1-p)\bm r_1$. A PVM with
Bloch axis $\bm n$ has cost
$C=\bar C-\frac12\max(|\gamma_0|,|\bm\gamma'\cdot\bm n|)$, where
$\bar C=\frac12[ap+b(1-p)]$ is the cost of a fair-coin guess. Its posteriors
satisfy
$\operatorname{sgn}(\eta_\pm-t_0)=\operatorname{sgn}(\gamma_0\pm\bm\gamma'\cdot\bm n)$.
The cost of the prior decision is $C^0=\bar C-\frac12|\gamma_0|=\min[ap,b(1-p)]$.
The spectrum straddles the threshold, $c-d<t_0<c+d$, if and only if the
Bayes-optimal cost $C^{\mathrm H}=\bar C-\frac12\max(|\gamma_0|,|\bm\gamma'|)$
is below $C^0$. That optimum is attained by the eigenbasis of $aA-bB$.
(d) Assume straddling, and let $\delta\theta$ be the angle between the axis
of the $I_f$-optimal PVM and the line of $\bm\gamma'$. Suppose that
$\delta\theta>\omega_c$, where
\begin{equation}
 \cos\omega_c=\frac{\bar C-C^0}{\bar C-C^{\mathrm H}}=\frac{|\gamma_0|}{|\bm\gamma'|}.
 \label{eq:gen-omega}
\end{equation}
Equivalently, both posteriors of the $I_f$-optimal PVM lie strictly on one
side of $t_0$, so that this measurement never changes the cost-optimal
decision. Then the best PVM must switch bases abruptly at some $\lambda_J>0$,
and $\Delta_{\lambda_J}>0$.
(e) In case (d), with both $\lambda=0$ contacts below $k$, the PVM optimal at
$\lambda=0$ stays optimal among all POVMs exactly for
$\bar\lambda\leq\min_{z\in(k,1]}D_0(z)/[2d(z-k)]$, and the gap opens there.
\end{theorem}

\par\textit{Proof}---
\begin{itemize}
 \item Lemma~\ref{lem:pareto-contacts} used only four facts:
 \begin{itemize}
  \item $D''=(\kappa-R)/s^3$, with $\{R>\kappa\}$ an interval;
  \item $g'$ is finite at $z=\pm1$, which holds because $f'$ is finite at
  $c\pm d\in(0,1)$;
  \item the slope drop $-2\bar\lambda d$ at $k$;
  \item strict concavity of $D$ when $\kappa=0$, which follows from strict
  concavity of $f$.
 \end{itemize}
 All four hold under (U), with $R\to R_f$ and $1/2\to t_0$.
 \item Lemma~\ref{lem:pareto-unique}, Theorem~\ref{thm:pareto-tie} and
 Corollary~\ref{cor:pareto-tie} then follow word for word.
 \item Part (c) is Lemma~\ref{lem:pareto-pvm-error} with $A-B$ replaced by
 $aA-bB=(a+b)\tau^{1/2}(Q-t_0I)\tau^{1/2}$.
 \item In Lemma~\ref{lem:pareto-plane}, $\bm\gamma'\in\Pi$; postprocessing cannot raise
 $I_f$ or lower $C$. Theorem~\ref{thm:pareto-forced} uses only part (c),
 uniqueness at $\lambda=0$, the unique cost-minimizing axis, and a finite
 bound on the score. In its proof the bound $h(p)/\lambda$ becomes $B_f/\lambda$, with
 $B_f=f(p)-\min_{[c-d,c+d]}f<\infty$, since $0\leq I_f\leq B_f$.
 \item Corollary~\ref{cor:pareto-onset} needs only $g_\lambda-g_0=\bar\lambda d|z-k|$.
\end{itemize}
For $p=t_0$, that is, when the prior sits at the threshold, one has
$\gamma_0=0$ and $\omega_c=\pi/2$, and the condition of (d) never holds. This
is the analogue of equal priors.

\textit{Existence.} The construction of Sec.~\ref{app:pareto-construct} needs
only (U), straddling, and $m_f\neq k$, where $m_f$ is the maximizer of $R_f$.
Under these conditions it gives a gap for an open set of average states. The
converse, no gap when $m_f=k$, is proved here only in the symmetric case
$f(\eta)=f(1-\eta)$, $t_0=c=1/2$. There $g_\lambda$ is even, so
$g_\lambda(z)=\varphi_\lambda(s(z))$, and $R_f$ is even, so $m_f=0=k$.
\begin{itemize}
 \item The error part, $\lambda d|z|=\lambda d\sqrt{1-s^2}$, is concave and
 decreasing in $s$.
 \item For the entropy part $g_0(z)=-f(\frac12+dz)=\varphi_0(s(z))$, $z>0$,
 differentiation gives $\varphi_0''(s)=[zs^2g_0''(z)-g_0'(z)]/z^3$.
 \item Now $g_0''=R_f/s^3$ and $g_0'(0)=0$, so $g_0'(z)=\int_0^zR_f(t)s(t)^{-3}\dd t$.
 \item Under (U) the even function $R_f$ decreases on $(0,1)$. With
 $\int_0^zs(t)^{-3}\dd t=z/s(z)$, this gives $g_0'(z)>R_f(z)z/s(z)=zs^2g_0''(z)$.
 \item Hence $\varphi_0''<0$. Also $\varphi_0'=-sg_0'/z<0$.
\end{itemize}
So $\varphi_\lambda$ is strictly concave and decreasing, and the argument of
Sec.~\ref{app:pareto-centered} for centered spectra gives one PVM that is optimal
for every $\lambda$. Symmetry alone, without (U), is not enough: the
Tsallis example below has $c=1/2$ and a gap at $\lambda=0$.

\prlsubheading{\protect\phantomsection\protect\label{app:tsallis-witness}Certified full-rank Tsallis counterexample}
Without (U), a standard entropy can require three outcomes even at
$\lambda=0$: the Tsallis entropy with $q=1/4$, $f(\eta)=\frac43[\eta^{1/4}+(1-\eta)^{1/4}-1]$, for
nearly pure states. Take $Q=\operatorname{diag}(0.999,0.001)$, that is $c=1/2$
and $d=0.499$, with $\tau=(I+0.5\sigma_x)/2$ and $p=1/2$. The following
certificate uses exact rationals and interval arithmetic (outward
rounding). In the $\tau$-frame convention
$M_y=t_y\tau^{-1/2}\sigma_y\tau^{-1/2}$, $\sigma_y=(I+\bm n_y\cdot\bm\sigma)/2$,
take $\bm n_{1,2}=(a_{\rm w},0,\pm b_{\rm w})$ and $\bm n_3=(1,0,0)$, where $a_{\rm w}=2u/(1+u^2)$,
$b_{\rm w}=(1-u^2)/(1+u^2)$ and $u=13013/10^6$. Take $t_{1,2}=1/[4(1-a_{\rm w})]$ and
$t_3=1-2t_1$. These rationals give $\sum_yt_y\sigma_y=\tau$ exactly and
$t_y\sigma_y\succeq0$ by rational principal minors, so the $M_y$ form a POVM.
Interval arithmetic gives $I_f\geq0.3403081851$ for it. Interval branch and
bound over the in-plane axis gives $I_f\leq0.3101758066$ for every PVM
($0.3101759$ without derivatives); the best PVM found has $I_f=0.310176$.
For this ensemble, out-of-plane PVMs and two-outcome POVMs have statistics
obtainable by postprocessing an in-plane PVM (Lemma~\ref{lem:pareto-plane}),
so they cannot do better. The gap is
therefore at least $0.030132378$ score units, and three outcomes are
necessary. The numerically optimized measurement has two nearly certain outcomes and one with
posterior $1/2$. This resembles the three-outcome pattern of unambiguous discrimination
\cite{ivanovic1987differentiate,dieks1988overlap,peres1988differentiate},
but full-rank signals have no error-free conclusive outcomes.
\prlsubheading{Relation to projective-reduction results}
\textit{No decision term ($\lambda=0$).} Under (U), the $I_f$-optimal
measurement of a noncommuting full-rank qubit pair is unique and is a PVM.
Writing $v(t)=pt+1-p$ and $t=\Tr\rho_0M_y/\Tr\rho_1M_y$, one has
$I_f=f(p)+\sum_y\Tr(\rho_1M_y)\,\varphi(t)$ with $\varphi(t)=-v(t)f(pt/v(t))$. This
is a measured $\varphi$-divergence \cite{reid2011information}. Projective sufficiency in every dimension
is known under the hypothesis of Theorem~2 of Ref.~\cite{fang2026uhlmann}: an
operator-concave $\psi$ with $\varphi^*\circ\psi$ operator convex. Hiai and Mosonyi ask whether an arbitrary operator-convex generator
always gives the same optimum over projective measurements and over all
POVMs \cite[Sec.~4.3, Eq.~(4.25)]{hiai2017different}.
 For the Shannon score, Sec.~\ref{app:v3-binary} derives projective
sufficiency from Theorem~2 of Ref.~\cite{fang2026uhlmann}. For the Brier score $f=2\eta(1-\eta)$,
$I_f=2p^2[\sum_y(\Tr\rho_0M_y)^2/\Tr\tau M_y-1]$ is an increasing function of
the measured R\'enyi divergence of order $2$ of $\rho_0$ from $\tau$, for which
projective measurements suffice by Theorem~4 of
Ref.~\cite{berta2017variational}. For these two scores only the uniqueness
and the trade-off statements are new here. We have not checked the
hypothesis of Ref.~\cite{fang2026uhlmann} for the Tsallis or R\'enyi
scores, or for the spherical score away from $p=1/2$ (see below), so for
them we claim only the qubit statement under (U), with uniqueness.

\prlsubheading{\protect\phantomsection\protect\label{app:operator-convex-score}Operator-convex generators and projective insufficiency}
For the Tsallis score with $0<q<1$, the corresponding $\varphi$ is operator convex on $(0,\infty)$.
Indeed, $(1-q)\varphi(t)=v-p^qt^qv^{1-q}-(1-p)^qv^{1-q}$. Here $v^{1-q}$ is an
operator-concave power of an affine function, and $t^qv^{1-q}$ is the
Kubo--Ando weighted geometric mean of the operator monotone functions $t$ and
$v(t)$. It is therefore operator monotone, and hence operator concave. The
example shows that projective measurements do not suffice for this
operator-convex $\varphi$. In terms of measured quantities,
$(1-q)[I_f-f(p)]=1-p^q\sum_y(\Tr\rho_0M_y)^q(\Tr\tau M_y)^{1-q}
-(1-p)^q\sum_y(\Tr\rho_1M_y)^q(\Tr\tau M_y)^{1-q}$. Each sum alone is a
measured R\'enyi quantity of order $q$, minimized by rank-one projective
measurements \cite[Thm.~4]{berta2017variational}, so the gap can arise only
because a single measurement must minimize their weighted sum. This is consistent with Theorem~2 of
Ref.~\cite{fang2026uhlmann}, whose hypothesis is not operator convexity of
$\varphi$ alone. It gives a faithful-state counterexample to projective
sufficiency for arbitrary operator-convex generators, the question posed in
Sec.~4.3 of Ref.~\cite{hiai2017different}. Sufficient
conditions for equality are known \cite[Thm.~5.8]{hiai2021quantum},
\cite[Thm.~2]{fang2026uhlmann}.
For two pure qubit states a negative answer already follows from
Theorem~2 of Ref.~\cite{katz2024optimal} together with the operator
convexity shown above: the Tsallis entropy of index $1/2$ equals
$2[(1+2b)^{1/2}-1]$ with $b=[\eta(1-\eta)]^{1/2}$, a concave function of
the Bhattacharyya objective $b$, so for priors that are not too skewed
its information is maximized by the three-outcome unambiguous
measurement. The present example provides an explicit faithful-state witness and an
interval certificate for the gap over all projective measurements, showing
that the failure does not rely on restricted supports. Its states are nearly pure,
with $\lambda_{\min}(\rho_i)=7.5\times10^{-4}$; the numerical scans below find
gaps for less pure pairs as well, without certificates.

\prlsubheading{\protect\phantomsection\protect\label{app:proper-score-reading}Proper-score bounds and their equality mechanism} Let $0<p<1$ and let $f$ be $C^1$ near
$J=[\lambda_{\min}(Q),\lambda_{\max}(Q)]$. Write
$I_f=f(p)+\sum_ip_i\sum_y\Tr(\rho_iM_y)\,s_i$, where each outcome is scored
by the Savage score of the true class's posterior,
$s_0(x)=-[f(x)+f'(x)(1-x)]$ at $x=\eta_y$ and
$s_1(x)=-[f(1-x)-(1-x)f'(1-x)]$ at $x=1-\eta_y$. Suppose that $s_0$ is
operator concave on $J$ and $s_1$ on $1-J$, a condition specific to the
ensemble. Then the operator Jensen inequality of Hansen and Pedersen
\cite{hansen2003jensen} bounds $I_f(M)$ by a concave function of
$B_M=\sum_y\eta_yM_y$, an operator built from the original effects, and
properness makes the eigenbasis of its maximizer optimal. This is Theorem~2
of Ref.~\cite{fang2026uhlmann} with the single choice $\psi=ps_0$,
localized to $J$. Their hypothesis is strictly broader: the spherical score,
whose Savage score is not even concave, fails operator concavity (although
it satisfies (U)), yet meets theirs at $p=1/2$ with $\psi=\mathrm{id}$ on
$[-1/2,0]$. To specify the conjugate domain, at $p=1/2$ use
$\varphi(t)=[\sqrt{1+t^2}-t-1]/2$ for $t\geq0$ and its tangent
extension $\varphi(t)=-t/2$ for $t<0$. Its conjugate is
$\varphi^*(z)=1/2-\sqrt{(-z)(1+z)}$ on $[-1/2,0]$ and is
infinite elsewhere. The square-root term is the geometric mean of two
positive affine operators, hence operator concave, so $\varphi^*$ is
operator convex. The extension leaves every positive likelihood ratio unchanged. Our contribution here is the interpretation in terms of a posterior
operator and a proper score.
Adding a decision term obstructs this proper-score argument. The score
$f+\lambda\min(a\eta,b(1-\eta))$ has a kink at $\eta=t_0$, so its Savage
score $s_0$ jumps there. Operator-concave functions are continuous, so
$s_0$ is not operator concave on any interval containing $t_0$, and
under straddling $t_0$ lies inside $J$.
A kink alone does not exclude projective optimality: Theorem~2 of
Ref.~\cite{fang2026uhlmann} covers the trace distance, whose generator has
a kink, and Helstrom's measurement is projective. The obstruction applies only to the argument using an operator-concave
proper score. For the logarithm, $s_0=\log$,
and this is the gap identity of Sec.~\ref{app:v3-binary}. For the Tsallis example above, write its index as $\nu=1/4$ and set
$g(r)=-r^{\nu-1}$ for the posterior-to-prior ratio $r$.
This $g$ is operator concave but is not a proper score. Its Jensen bound is
\[
 U_g=\frac{\displaystyle\sum_i p_i^\nu+
 \max_H\sum_i p_i^\nu\Tr\rho_i g(G_i(H))}{1-\nu},
\]
where $G_0(H)=I+(1-p)H$, $G_1(H)=I-pH$, and the maximum is over
$G_0(H),G_1(H)\succ0$, as in Sec.~\ref{app:v3-binary}.
The identity $\eta_i=p_i r_i$ and operator Jensen give $I_f(M)\leq U_g$.
For this instance, $U_g=0.4727$, whereas the eigenbasis of its maximizer
gives only $I_f=0.2518$ and the certified three-outcome POVM gives
$0.3403$ (score units). Thus the bound is valid but not attained:
commutation alone does not supply the proper-score gap identity.
The certified gap also excludes any representation that is both proper
and operator concave on this ensemble's posterior intervals.

\prlsubheading{A smooth deformation of the Shannon task}
Condition (U) cannot simply be dropped. Take
$f_\epsilon=h+\lambda_0m_\epsilon$ with
$m_\epsilon(\eta)=\frac12-\sqrt{(\eta-\frac12)^2+\epsilon^2}$. This $f_\epsilon$ is
smooth and strictly concave. For sufficiently small $\epsilon$, $R_{f_\epsilon}$
is bimodal: for the spectrum of Fig.~\ref{fig:pareto-gap} ($c=0.30$,
$d=0.29$) with $\lambda_0=0.51046$, this holds numerically for
$\epsilon\leq0.08$, and $R_{f_\epsilon}$ is unimodal at $\epsilon=0.1$. Two facts give a bound on the
change in the gap. First, $f_\epsilon$ differs from
$h+\lambda_0\min(\eta,1-\eta)$ by at most $\lambda_0\epsilon$, so each
information functional, and hence each optimized value, moves by at most
$2\lambda_0\epsilon$. Second, $I_{h+\lambda_0\min}=F_{\lambda_0}+\lambda_0\Pe^{0}$.
Hence $|\Delta_{f_\epsilon}-\Delta_{\lambda_0}|\leq4\lambda_0\epsilon$. Take the
ensemble of Fig.~\ref{fig:pareto-gap} at $\lambda_0=0.51046$, near the switch.
The value $\Delta_{\lambda_0}\simeq2.308\times10^{-3}$ nat then gives $\Delta_{f_\epsilon}>0$ for
$\epsilon<1.1\times10^{-3}$, and $I_{f_\epsilon}$ is not maximized by any orthogonal
measurement of this qubit pair.

Two parts of the chain survive without (U). The forced switch of (d) needs
only a unique $I_f$-optimal PVM, since the argument of
Theorem~\ref{thm:pareto-forced} never uses POVMs. The gap at the switch, part
(b), uses the contact count and so requires (U); without (U) it is only
observed numerically, below.

\prlsubheading{Numerical checks and scope without (U)}
The scan of
Sec.~\ref{app:pareto-size} was repeated with random pairs in the
Hilbert--Schmidt ensemble. In every run the two forms of condition (d) agree
in every straddling pair, and Eq.~\eqref{eq:gen-omega} matches the angle form
to $2\times10^{-15}$. In every run with (U), no pair has $\Delta_0>10^{-9}$.

The Tsallis advantage at $\lambda=0$ is not isolated. For the Tsallis spectrum $Q=\operatorname{diag}(0.999,0.001)$ above, 112 of 188 average states on
a grid have $\Delta_0>10^{-9}$, up to $4.4\times10^{-2}$ score units. Among 3834
straddling random pairs, 43 do, up to $1.3\times10^{-2}$. Tsallis $q=1/2$
and R\'enyi $\alpha=1/2$ give similar results; the Shannon entropy on the same
grids gives none. Grid: $\tau=(x,0,z)$, with $x$ and $z$ taking 16 equally
spaced values in $[0.02,0.98]$ and $[-0.98,0.98]$ and $x^2+z^2<0.97$ (188
states). Random pairs: $p$ uniform on $(0,1)$ and Bloch vectors uniform in
the ball, drawn with NumPy's default generator with seeds 101--105 in the
row order of Table~\ref{tab:gen-scan}, and kept if $c-d<t_0<c+d$. A gap
means $\Delta>10^{-9}$. Each $\Delta$ is the value of an explicit POVM minus
a refined maximum over 8192 PVM chords, a refined grid search rather than a
certificate. The LP dual brackets are at most $2\times10^{-8}$ for the values
at $\lambda=0$ and at most $1.1\times10^{-13}$ for the gaps at the switches
in Table~\ref{tab:gen-scan}. Switches come from a hull over 16384 chords,
with the tie located by numerical root finding.

\begin{table}[htbp]
\caption{Switches and gaps for general losses. Columns: the information score
$f$; the cost threshold $t_0$; the number of pairs; the number straddling
$t_0$; the number satisfying condition (d); the number whose best basis
switches; the smallest $\Delta_{\lambda_J}$ over the switches, in nat for the
Shannon entropy and in score units otherwise. Every switch has three
outcomes and a gap. LP brackets are at most $1.1\times10^{-13}$ in the same
units.}
\label{tab:gen-scan}
\begin{ruledtabular}
\begin{tabular}{lcrrrrc}
$f$ & $t_0$ & pairs & straddling & (d) & switches & $\min\Delta_{\lambda_J}$\\
Shannon & 0.30 & 8000 & 4693 & 461 & 461 & $1.1\times10^{-8}$\\
Shannon & 0.15 & 4000 & 1814 & 282 & 282 & $2.8\times10^{-7}$\\
Brier & 0.50 & 8000 & 5088 & 222 & 222 & $1.1\times10^{-9}$\\
spherical & 0.65 & 6000 & 3684 & 156 & 156 & $3.8\times10^{-8}$\\
Tsallis $q=\frac14$ (no (U)) & 0.50 & 6000 & 3834 & 738 & 788 & $2.5\times10^{-8}$\\
\end{tabular}
\end{ruledtabular}
\end{table}
With (U), condition (d) and the switch coincide in every pair. For Tsallis
$q=1/4$, every pair satisfying (d) switches, 50 pairs switch without it, and
all 788 switches have a gap.

\end{document}